\documentclass[11pt,a4paper]{article}
\usepackage{bbding}
\fussy
\usepackage[pagebackref=false,citecolor=newblue,colorlinks=true,linkcolor=newblue,bookmarks=false]{hyperref}

\usepackage[font=scriptsize,bf]{caption}
\usepackage{epsfig,amssymb,amsfonts,amsmath,amsthm}
\usepackage[margin=1in]{geometry}

\usepackage[ruled,vlined, linesnumbered]{algorithm2e}
\usepackage[noend]{algpseudocode}

\SetKwInput{KwInput}{Input}                
\SetKwInput{KwOutput}{Output}              

\usepackage{amsopn}

\usepackage{xcolor}
\usepackage{rotating}
\usepackage{enumerate}
\usepackage{enumitem}

\usepackage{mdframed}

\usepackage{wrapfig}
\usepackage{subcaption}
\usepackage{tikz}
\usetikzlibrary{shapes.geometric,positioning}
\usetikzlibrary{arrows.meta}

\definecolor{lime}{HTML}{A6CE39}
\DeclareRobustCommand{\orcidicon}{%
	\begin{tikzpicture}
	\draw[lime, fill=lime] (0,0) 
	circle [radius=0.16] 
	node[white] {{\fontfamily{qag}\selectfont \tiny ID}};
	\draw[white, fill=white] (-0.0625,0.095) 
	circle [radius=0.007];
	\end{tikzpicture}
	\hspace{-2mm}
}

\foreach \x in {A, ..., Z}{%
	\expandafter\xdef\csname orcid\x\endcsname{\noexpand\href{https://orcid.org/\csname orcidauthor\x\endcsname}{\noexpand\orcidicon}}
}

\numberwithin{equation}{section}

\definecolor{ocre}{rgb}{0.72,0,0}
\definecolor{MyMango}{rgb}{1.00, 0.47, 0.20}
\definecolor{brickred}{rgb}{0.8, 0.25, 0.33}
\definecolor{newblue}{rgb}{0.2,0.2,0.6} 

\definecolor{babyblueeyes}{rgb}{0.63, 0.79, 0.95}
\definecolor{newgreen}{rgb}{0.53,0.66,0.42}
\definecolor{newred}{rgb}{0.67,0.16,0}

\definecolor{forestgreen}{rgb}{0.13, 0.55, 0.13}
\definecolor{applegreen}{rgb}{0.55, 0.71, 0.0}
\definecolor{amethyst}{rgb}{0.6, 0.4, 0.8}

\newcommand{\Deg}{\mathbf{D}}
\newcommand{\Adj}{\mathbf{A}}
\newcommand{\NLap}{\mathbf{\mathcal{L}}}
\newcommand{\davg}{d_{\mathrm{avg}}}
\newcommand{\dmin}{d_{\min}}
\newcommand{\dmax}{d_{\max}}

\newcommand{\vol}{\mathrm{vol}}

\newcommand{\phiIn}{\phi_{\mathrm{in}}}
\newcommand{\phiOut}{\phi_{\mathrm{out}}}

\newcommand{\lp}{\left (}
\newcommand{\rp}{\right )}

\newcommand{\rsp}{\right ]}

\newtheorem{definition}{Definition}[section]

\newtheorem{lemma}[definition]{Lemma}

\newtheorem{claim}[definition]{Claim}
\newtheorem{theorem}{Theorem}

\newtheorem{example}[definition]{Example}

\newcommand{\leaves}{\mathsf{leaves}}
\newcommand{\OPT}{\mathsf{OPT}}
\newcommand{\cost}{\mathrm{cost}}

\newcommand{\parent}{\mathrm{parent}}

\newcommand{\T}{\mathcal{T}}

\newcommand{\TPrune}{\T_{\mathrm{PM}}}
\newcommand{\core}{\mathrm{core}}

\newcommand{\poly}{\ensuremath{\mathrm{poly}}}
\newcommand{\barr}{\overline}
\providecommand{\abs}[1]{\lvert#1\rvert}

\newcommand{\ceil}[1]{\left \lceil #1 \right \rceil}
\newcommand{\imin}{i_{\min}}
\newcommand{\imax}{i_{\max}}

\newcommand{\SpecPart}{\texttt{Spectral Partitioning} }

\allowdisplaybreaks

\title{Hierarchical Clustering: \\ $O(1)$-Approximation for Well-Clustered Graphs\footnote{A preliminary version of the work  appeared at the 35th Conference on Neural Information Processing Systems~(NeurIPS~'21).}}  

\author{Bogdan-Adrian Manghiuc\footnote{School of Informatics,    University of Edinburgh, UK. \url{b.a.manghiuc@sms.ed.ac.uk}. This work is supported by an EPSRC Doctoral Training Studentship (EP/R513209/1). \orcidA}
\and He Sun\footnote{School of Informatics, University of Edinburgh, UK. \url{h.sun@ed.ac.uk}. This work is supported by an EPSRC Early Career Fellowship~(EP/T00729X/1).}  }
\date{}

\begin{document}
\maketitle
 
\thispagestyle{empty}

\setcounter{page}{0}

\begin{abstract}
Hierarchical clustering  studies a recursive partition of a data set into clusters of successively smaller size, and is a fundamental problem in data analysis. 
In this work we study the cost function for hierarchical clustering introduced by Dasgupta~\cite{dasgupta2016cost}, and present two polynomial-time approximation algorithms:  
Our first result is an $O(1)$-approximation algorithm for  graphs of high conductance. Our simple construction bypasses complicated recursive routines of finding sparse cuts known in the literature~(e.g., \cite{cohen2018hierarchical,charikar2017approximate}).
Our second and main result is an $O(1)$-approximation algorithm for a wide family of graphs that exhibit a well-defined structure of clusters.
This result generalises the previous state-of-the-art~\cite{CAKMT17}, which holds only for graphs generated from  stochastic models. The significance of our work is demonstrated by the   empirical analysis on both synthetic and real-world data sets, on which our presented algorithm outperforms  the previously proposed algorithm for graphs with a well-defined cluster structure \cite{CAKMT17}.
\end{abstract}

\newpage
\thispagestyle{empty}
\setcounter{page}{0}
\tableofcontents
\newpage

\section{Introduction}
Hierarchical clustering~(\textsf{HC}) studies a recursive partition of a data set into clusters of successively smaller size, via an effective binary tree representation. As a basic technique, hierarchical clustering has been employed as a standard package in  data analysis, and has comprehensive applications in practice. 
While traditionally \textsf{HC} trees are constructed through bottom-up (agglomerative) heuristics, which lacked a clearly-defined objective function, Dasgupta~\cite{dasgupta2016cost} has recently introduced a simple objective function to measure the quality of a particular hierarchical clustering and his work has inspired a number of research on  this topic~\cite{alon2020hierarchical,CAKMT17, cohen2018hierarchical,charikar2017approximate,charikar2019better_than_AL, charikar2019euclidean,  moseley2017approximation, roy2017hierarchical}. 
Consequently, there has been a significant interest in studying efficient \textsf{HC} algorithms that not only work in practice, but also have proven theoretical guarantees with respect to Dasgupta's cost function.

\subsection{Our contribution}  

We present two new approximation algorithms for constructing \textsf{HC} trees that can be rigorously analysed with respect to Dasgupta's cost function.
For our first result, we construct an  \textsf{HC} tree of an input graph $G$  \emph{entirely} based on the degree sequence of $V(G)$, and   show that the approximation guarantee of our constructed tree is with respect to the conductance of $G$, which will be defined formally in Section~\ref{sec:pre}. 
The striking fact of this result is that, for any $n$-vertex graph $G$ with $m$ edges and conductance $\Omega(1)$~(a.k.a. expander graph), an $O(1)$-approximate \textsf{HC} tree of $G$ can be very easily constructed in  $O(m + n\log n)$ time, although obtaining such result for general graphs in polynomial time is impossible under the Small Set Expansion Hypothesis~(\textsf{SSEH})~\cite{charikar2017approximate}. 
Our theorem is in line with a sequence of results for problems that are naturally linked to the Unique Games and Small Set Expansion problems: it has been shown that such problems are much easier to solve once the input instance exhibits the high  conductance property~\cite{ABS15,AKKSTV08,Kolla10,Li0Z19}.
However, to the best of our knowledge, our result
is the first of this type for hierarchical clustering, and can be
 informally described as follows:

\begin{theorem}[informal statement of Theorem~\ref{thm:degree}]\label{thm:degree_informal}
Given any graph $G=(V,E,w)$ with constant  conductance  as input, there is an algorithm that  runs in $O(|E| + |V| \log (|V|))$ time and returns an $O(1)$-approximate \textsf{HC} tree of $G$.
\end{theorem}

While our first result presents an interesting theoretical fact, we further study whether we can extend this $O(1)$-approximate construction to a much wider family of graphs occurring in practice. 
Specifically, we look at \emph{well-clustered graphs}, i.e.,  the graphs in which vertices within each cluster are better connected than vertices between different clusters  and the  total number of clusters is constant. This includes a wide range of graphs occurring in practice with a clear cluster-structure, and have been extensively studied over the past two decades~(e.g., \cite{GT14,KannanVV04,PengSZ17, SZ19, Luxburg07}). 
As our second and main result, we present an approximation   algorithm   for  well-clustered graphs, and our result is informally described as follows:

\begin{theorem}[informal statement of Theorem~\ref{thm:main_k_clusters}]\label{thm:main_k_clusters_informal}
Let $G=(V, E, w)$ be a well-clustered graph with $O(1)$ clusters. Then, there is a polynomial-time algorithm that constructs an $O(1)$-approximate \textsf{HC} tree of $G$.
\end{theorem} 

Given that the class of well-clustered graphs includes graphs with clusters of different sizes and  asymmetrical internal structure, our result significantly improves the previous state-of-the-art \cite{CAKMT17}, which only holds for graphs generated from stochastic models.
At the technical level, the design of our algorithm is based on the graph decomposition algorithm presented in \cite{GT14}, which is designed to find a good partition of a well-clustered graph.
However, our analysis suggests that, in order to obtain an $O(1)$-approximation algorithm, directly applying their decomposition is not sufficient for our purpose. To overcome this bottleneck, we refine their output decomposition via a \emph{pruning} technique, and carefully merge the refined parts to construct our final $\textsf{HC}$ tree.
In our point of view, our presented stronger graph decomposition procedure might have applications in other settings as well.

To demonstrate the significance  of our work, we  compare our algorithm against the previous state-of-the-art with similar approximation guarantee~\cite{CAKMT17}  and well-known linkage heuristics on both synthetic and real-world data sets. Although our algorithm's performance is marginally better than \cite{CAKMT17} for the graphs generated from the stochastic block models~(\textsf{SBM}), the cost of our algorithm's output is  up to  $50\%$ lower than the one from \cite{CAKMT17} when the clusters of the input graph have different sizes and some cliques are embedded into a cluster. 

\subsection{Related work}

Our work fits in a line of research initiated by Dasgupta~\cite{dasgupta2016cost}, who introduced a cost function to measure the quality of an \textsf{HC} tree.  
Dasgupta proved that a recursive application of the algorithm for the Sparsest Cut problem
can be used to construct an $O(\log^{3/2}n)$-approximate \textsf{HC} tree. The approximation factor was first improved to $O(\log n)$ by 
Roy and Pokutta~\cite{roy2017hierarchical}.
Charikar and Chatziafratis~\cite{charikar2017approximate} improved  Dasgupta's analysis of the recursive sparsest cut algorithm by establishing the following black-box connection: an $\alpha$-approximate algorithm  for the Sparsest Cut problem can be used to construct  an $O(\alpha)$-approximate \textsf{HC} tree according to Dasgupta's cost function. Hence,  an $O(\sqrt{\log n})$-approximate \textsf{HC} tree can be computed in polynomial time by using the celebrated result of \cite{ARV09}.
For general input instances, it is known to be \textsf{NP}-hard to find an optimal \textsf{HC} tree~\cite{dasgupta2016cost}, and \textsf{SSEH}-hard to achieve an $O(1)$-approximation with respect to Dasgupta's cost function~\cite{roy2017hierarchical, charikar2017approximate}.

Cohen-Addad et al.~\cite{cohen2018hierarchical} analysed the performance of several linkage heuristics~(e.g.,~\texttt{Average Linkage}) for constructing \textsf{HC} trees, and showed that such algorithms could  output an \textsf{HC}  tree of   high cost in the worst case. 
Moving beyond the worst-case scenario, Cohen-Addad et al.~\cite{CAKMT17} studied a hierarchical extension of the \textsf{SBM} and showed that, for   graphs generated    according to this model, a certain SVD projection algorithm~\cite{mcsherry2001spectral} together with several linkage heuristics can be applied to construct  a $(1+o(1))$-approximate \textsf{HC} tree with high probability. We emphasise that our notion of well-clustered graphs generalises the \textsf{SBM} variant studied in \cite{CAKMT17}, and does not assume the rigid hierarchical structure of the clusters.

For another line of related work,  Moseley and Wang~\cite{moseley2017approximation} studied the dual objective function and proved that \texttt{Average Linkage} achieves a $(1/3)$-approximation for the new  objective. Notice that, although this  has received significant  attention very recently~\cite{alon2020hierarchical,  charikar2019better_than_AL, charikar2019euclidean, chatziafratis2020bisect,  vainstein2021hierarchical}, achieving an $O(1)$-approximation is tractable under this alternative objective. This suggests the fundamental difference on the hardness of the problem under different objective functions,  and is our reason  to entirely focus on Dasgupta's cost function in this work. 

\subsection{Organisation}

The remaining part of the paper is organised as follows: we introduce the necessary notation about graphs and matrices, the basis of hierarchical clustering and graph partitioning in Section~\ref{sec:pre}. 
In Section~\ref{sec:One Expander}, we present and analyse the algorithm of constructing an \textsf{HC} tree for graphs of high conductance. 
We present and analyse our main algorithm, i.e., the algorithm for constructing $O(1)$-approximate \textsf{HC} trees for well-clustered graphs in Section~\ref{sec:Well-Clustered Graphs}.
In Section~\ref{sec:Experiments} we present the experimental analysis where we compared our developed algorithm against other algorithms in the literature. We end the paper with several concluding remarks and directions of future work in Section~\ref{sec:Conclusion}.  

\section{Preliminaries}\label{sec:pre}


Throughout the paper, we always assume that $G=(V,E,w)$ is an undirected graph with $|V| = n$ vertices, $|E| = m$ edges and weight function $w: V\times V\rightarrow \mathbb{R}_{\geq 0}$. 
For any edge $e = \{u, v\} \in E$, we write $w_e$ or $w_{uv}$ to indicate the \emph{similarity} weight between $u$ and $v$. 
For a vertex $u \in V$, we denote its  \emph{degree}  by $d_u \triangleq \sum_{v \in V} w_{uv}$ and we assume that $w_{\mathrm{max}}/w_{\mathrm{min}} = O(\poly(n))$, where $w_{\mathrm{min}}(w_{\mathrm{max}})$ is the minimum (maximum) edge weight.  
We will use $\dmin, \dmax$, and $\davg$ for the minimum, maximum and average degrees in $G$ respectively, where $\davg \triangleq \sum_{u \in V} d_u / n$.
For a nonempty subset $S\subset V$, we define $G[S]$ to be the induced subgraph on $S$ and we denote by $G\{S\}$ the subgraph $G[S]$, where self loops are added to vertices $v \in S$ such that their degrees in $G$ and $G\{S\}$ are the same.
For any two subsets $S, T \subset V$, we define the \emph{cut value} $w(S, T) \triangleq \sum_{e \in E(S, T)} w_e$, where $E(S, T)$ is the set of edges between $S$ and $T$.
For any $G=(V,E, w)$ and  set $S \subseteq V$, the \emph{volume} of $S$ is $\vol_G(S)\triangleq \sum_{u\in S} d_u$,
and we write $\vol(G)$ when referring to $\vol_G(V)$. Sometimes we drop the subscript $G$ when it is clear from the context.  For any nonempty subset $S \subseteq V$, 
we define the \emph{conductance} of $S$ by 
\[
\Phi_G(S) \triangleq \frac{w(S, V\setminus S )}{\vol(S)}.
\]
 Notice that $\Phi_G(V) = 0$ and we conventionally choose $\Phi_G(\emptyset) = 1$, where $\emptyset$ is the empty set.
 Furthermore, we define the conductance of the graph $G$ by
\[
 \Phi_G\triangleq \min_{\substack{S\subset V\\ \vol(S) \leq \vol(V)/2}} \Phi_G(S),
\]
and we call $G$ an \emph{expander} graph if $\Phi_G = \Omega(1)$.

For a graph $G = (V, E, w)$,  let $\Deg\in\mathbb{R}^{n\times n}$ be the diagonal matrix defined by $\Deg_{uu} = d_u$ for all $u \in V$. We  denote by  $\Adj\in\mathbb{R}^{n\times n}$  the  \emph{adjacency matrix}  of $G$, where $\Adj_{uv} = w_{uv}$ for all $u, v \in V$. The \emph{normalised Laplacian matrix} of $G$ is defined as $\NLap \triangleq \mathbf{I} - \Deg^{-1/2} \Adj \Deg^{-1/2}$, where $\mathbf{I}$ is the $n \times n$ identity matrix.  The normalised Laplacian $\NLap$ is symmetric and real-valued, therefore it has $n$ real eigenvalues which we will write as  $\lambda_1 \leq \ldots \leq \lambda_n$. It is known that $\lambda_1 = 0$ and $\lambda_n \leq 2$~\cite{chung1997spectral}.

\subsection{Hierarchical clustering}

A \emph{hierarchical clustering (\textsf{HC}) tree} of a given graph $G$ is a binary tree $\T$ with $n$ leaf nodes such that each leaf corresponds to exactly one vertex $v \in V(G)$.
Let $\T$ be an \textsf{HC} tree of some graph $G = (V, E, w)$, and  $N \in \T$ be an arbitrary internal node\footnote{We consider any non-leaf node of $\mathcal{T}$ an \emph{internal node}. We will always use the term \emph{node(s)} for the nodes of $\T$ and the term \emph{vertices} for the elements of the vertex set $V$.} of $\T$. 
We write $N \rightarrow (N_1, N_2)$ to indicate that $N_1$ and $N_2$ are the children of $N$. 
We denote $\T[N]$ to be the subtree of $\T$ rooted at $N$,  $\leaves \lp \T[N]\rp$ to be the set of leaf nodes of $\T[N]$ and  $\parent_{\T}(N)$ to be the parent of node $N$ in $\T$.
 In addition, each internal node $N\in \T$   induces a unique vertex set $C \subseteq V$ formed by the vertices corresponding to $\leaves(\T[N])$.  
For the ease of presentation, we will sometimes abuse the notation,   and write $N \in \T$ for both the internal node of $\T$ and the corresponding subset of vertices in $V$.

To measure the quality of an \textsf{HC} tree $\T$ with similarity weights, Dasgupta \cite{dasgupta2016cost} introduced the   cost function defined by
\[
    \mathrm{cost}_{G}(\T) \triangleq \sum_{e = \{u,v\} \in E} w_e \cdot \abs{\leaves \lp \T[u \vee v]\rp},
\]
where $u \vee v$ is the lowest common ancestor of $u$ and $v$ in $\T$. 
Trees that achieve a better hierarchical clustering have a lower cost, and the objective of \textsf{HC} is to construct trees with the minimum cost based on the following consideration: for any pair of vertices $u, v \in V$ that corresponds to an edge of high weight $w_{uv}$ (i.e., $u$ and $v$ are highly similar) a ``good'' \textsf{HC} tree would  separate $u$ and $v$ lower in the tree, thus reflected in a small size $|\leaves(\T[u \vee v])|$. 
We denote by $\OPT_G$ the minimum cost of any \textsf{HC} tree of $G$, i.e., $\OPT_G = \min_{\T} \mathrm{cost}_G(\T)$, and use the notation $\T^*$ to refer to  an \emph{optimal} tree achieving the minimum cost. We say that an \textsf{HC} tree $\T$ is an \emph{$\alpha$-approximate} tree if $\cost_G(\T)\leq \alpha\cdot \OPT_G$, for some $\alpha \geq 1$.

 Sometimes, it is convenient to consider the cost of an edge $e = \{u, v\} \in E$ in $\T$ as  
\[
    \cost_{\T}(e) \triangleq w_e \cdot \abs{\leaves(\T[u \vee v])},
\]
so that we can write 
\[
    \cost_{G}(\T) = \sum_{e\in E} \cost_{\T}(e).
\]

Alternatively, as observed by Dasgupta~\cite{dasgupta2016cost}, the cost function can be expressed with respect to all cut values induced at every internal node.

\begin{lemma}[\cite{dasgupta2016cost}]
    The cost function of an \textsf{HC} tree $\T$ of $G$ can be written as
    \[
        \cost_G(\T) = \sum_{\substack{N \in \T \\ N \rightarrow (N_1, N_2)}} \abs{\leaves(\T[N])} \cdot w(\leaves(\T[N_1]), \leaves(\T[N_2])),
    \]
    where the sum is taken over all internal nodes $N$.
\end{lemma}
 
 The following lemma presents a simple upper bound on the cost of any \textsf{HC} tree $\T$:
\begin{lemma}\label{lem:Trivial bound}
 It holds for    any \textsf{HC} tree $\T$ of $G$   that 
    \[
        \mathrm{cost}_G(\T) \leq \frac{n \cdot \vol(G)}{2}.
    \]
\end{lemma}
\begin{proof}
    For any \textsf{HC} tree $\T$ of $G$, we  have that 
    \[
    \cost_G(\T) = \sum_{e \in E} \cost_{\T}(e) \leq \sum_{e \in E} n \cdot w_e \leq \frac{n \cdot \vol(G)}{2}, 
    \]
    which proves the statement.
\end{proof}

\subsection{Graph partitioning}

The following results on graph partitioning will be used in our analysis, and we list them here for completeness.

\begin{lemma}[Cheeger Inequality, \cite{Alon/86}]\label{lem:Cheeger's ineq}
It holds for any graph $G$ that 
    \[
        \frac{\lambda_2}{2} \leq \Phi_G \leq \sqrt{2\lambda_2}.
    \]
    Furthermore, there is a  nearly-linear time algorithm\footnote{We say that a graph algorithm runs in nearly-linear time if the algorithm runs in $O(m \cdot\mathrm{poly}\log n)$ time, where $m$ and $n$ are the number of edges and vertices of the input graph.} (i.e., the \SpecPart algorithm) that finds a set $S$ such that $\vol(S) \leq \vol(V)/2$, and $\Phi_G(S) \leq 2\cdot \sqrt{ \Phi_G}$. 
\end{lemma}

One can generalise the notion of conductance, and for any $k\geq 2$ define the \emph{$k$-way expansion} of $G$ by
\begin{equation*}
    \rho(k) \triangleq \min_{\mathrm{disjoint \:} S_1, \dots, S_k} \max_{1\leq i \leq k} \Phi_G(S_i).
\end{equation*}

\begin{lemma}[Higher-Order Cheeger Inequality,  \cite{higherCheeg}]\label{lem:Higher Cheeger}
It holds for any graph $G$ and $k \geq 2$  that
    \[
        \frac{\lambda_k}{2} \leq \rho(k) \leq O(k^2) \sqrt{\lambda_k}.
    \] 
\end{lemma}

\begin{lemma}[Lemma~1.13, \cite{GT14}]\label{lem:Upperbound eig induced graphs}
    There is a universal constant $c_0 > 1$ such that for any $k\geq 2$ and any partitioning of $V$ into $r$ sets $P_1, \dots P_{r}$ of $V$, where $r \leq k - 1$, we have that 
    \[
        \min_{1 \leq i \leq r} \lambda_2(G[P_i]) \leq 2 c_0 \cdot k^6 \cdot \lambda_{k}.
    \]
\end{lemma}

\section{Hierarchical clustering for graphs of high conductance}\label{sec:One Expander}
In this section we study hierarchical clustering for graphs with high conductance and prove that, for any input graph $G$ with $\Phi_G=\Omega(1)$, an $O(1)$-approximate \textsf{HC} tree of $G$ can be simply constructed based on the degree sequence of $G$. This section is organised as follows: in Section~\ref{sec:Deg Upper Bound}, we give an upper bound for the cost of any \textsf{HC} tree $\T$ based on the degree distribution and conductance of $G$. In Section~\ref{sec:Alg for expanders}, we present an algorithm  for  constructing an \textsf{HC} tree for graphs of high conductance, whose analysis is presented in Section~\ref{sec:Analysis expanders}.

\subsection{Upper bounding  $\cost_G(\T)$ with respect to the degrees of $G$}\label{sec:Deg Upper Bound}

As a starting point, we show that   $\cost_G(\T)$ for any $\T$ can be upper bounded with respect to $\Phi_G$ and the degree distribution of $V(G)$.

\begin{lemma}\label{lem:simplebound}
It holds for any    \textsf{HC} tree $\T$ of graph $G$ that 
$
    \mathrm{cost}_{G}(\T) \leq  \frac{9}{4\Phi_G} \cdot \min \left\{\frac{d_{\mathrm{avg}}}{d_{\min}}, \frac{d_{\max}}{d_{\mathrm{avg}}} \right\} \cdot \OPT_G.$
\end{lemma}

The proof is based on a combination of Lemma~\ref{lem:Trivial bound} and the following technical result:
  
\begin{lemma}\label{lem:Lower Bound Cost Expander}
    It holds for any optimal \textsf{HC} tree $\T^*$ of graph $G$ that
\[
    \mathrm{cost}_G(\T^*) \geq \frac{2\Phi_G}{9} \cdot \max \left\{\frac{\vol(G)^2}{d_{\max}} , d_{\min} \cdot n^2 \right\} = \frac{2\Phi_G}{9} \cdot n\cdot \vol(G) \cdot \max \left\{ \frac{d_{\mathrm{avg}}}{d_{\max}}, \frac{d_{\min}}{d_{\mathrm{avg}}}\right\}.
\]
\end{lemma}
 We remark that, as a corollary, the lower bound above holds for any \textsf{HC} tree $\T$ of $G$.  
 
\begin{proof}[Proof of Lemma~\ref{lem:Lower Bound Cost Expander}]
    We will give two lower bounds for $\mathrm{cost}_G(\T^*)$. Let $A_0$ be the root of $\T^*$. For the first lower bound, we start with the root node $A_0$ and travel along $\T^*$ recursively as follows: at every intermediate node $A_i$, we travel down to the node of higher \emph{volume} among its two children. This process stops when we reach node $A_k$ such that $\vol(A_k) \leq \frac{2 \vol(G)}{3}$. We denote $ \lp A_0, \dots, A_k \rp$ to be the path in $\T^*$ from the root $A_0$ to $A_k$ and we define $A\triangleq A_k$ as well as $B\triangleq V\setminus A_k$. By construction, it holds that $\vol(A)> \vol(G)/3$ and $\vol(B)\geq\vol(G)/3$. We will show that the cut $(A,B)$ has significant contribution to $\cost_G\left(\T^*\right)$.

    By the stopping criteria, we know that $\vol(A_i) > \frac{2\vol(G)}{3}$ for every $0\leq i < k$. 
    For any edge $e =\{u, v\}$ in the cut $(A,B)$, if $A_i = u \vee v$ for some $0 \leq i < k$, then we have that 
    \[
        \abs{\leaves \lp \T^*[ u \vee v] \rp} = |A_i| = \frac{\vol(A_i)}{d_\mathrm{avg}(A_i)} > \frac{2\cdot \vol(G)}{3\cdot d_{\mathrm{max}}}.
    \]
    Therefore, we have that 
    \begin{align*}
        \mathrm{cost}_G(\T^*) 
        &= \sum_{e = \{u, v\}} w_e \cdot \abs{\leaves\lp \T^* [u \vee v]\rp}
        \geq \sum_{\substack{ e\in E(A, B) \\ e=\{u,v\}}} w_e \cdot \abs{\leaves\lp \T^* [u \vee v]\rp}\\
        &\geq \frac{2}{3}\cdot \frac{\vol(G)}{d_\mathrm{max}} \cdot w(A, B)\geq \frac{2}{3}\cdot \frac{\vol(G)}{d_\mathrm{max}} \cdot \Phi_{G} \cdot \min\{\vol(A), \vol(B)\}\\
        &\geq \frac{2}{3}\cdot \frac{\vol(G)}{d_\mathrm{max}} \cdot \Phi_G \cdot \frac{\vol(G)}{3}\\
        &= \frac{2\Phi_G \cdot \vol(G)^2}{9\cdot d_\mathrm{max}}.
    \end{align*} 
    
    The second lower bound is proven in a similar way: we start with the root node $A_0$, and travel along $\T^*$ recursively as follows: at every intermediate node $A_i'$ we travel down to the node of larger \emph{size}; this process stops when we reach node $A'_{\ell}$ such that $|A^{'}_{\ell}|\leq 2n/3$. We denote $(A_0=A_0',\ldots, A^{'}_{\ell})$ to be the path in $\T^*$ from the root to $A_\ell'$, and we define $A^{'}\triangleq A^{'}_{\ell}$ as well as  $B^{'}\triangleq V\setminus A^{'}_{\ell}$. By construction, it holds that $|A^{'}|> n/3$, and $|B^{'}|> n/3$. We will show that the cut $(A^{'}, B^{'})$  has significant contribution to $\cost_{G}(\T^{*})$. 
    
Similar to the analysis in the first case,  by the stopping criteria we have that $|A'_i| > \frac{2n}{3}$, for all $0 \leq i < \ell$. Moreover, for any $e = (u,v) \in E(A', B')$, we have that $|\leaves(\T^*[u \vee v])| > \frac{2n}{3}$. Hence, it holds that 
    \begin{align*}
        \mathrm{cost}_G(\T^*) 
        &= \sum_{e = \{u, v\}} w_e \cdot \abs{\leaves\lp \T^* [u \vee v]\rp}
        \geq \sum_{e = \{u, v\} \in E(A', B')} w_e \cdot \abs{\leaves\lp \T^* [u \vee v]\rp}\\
        &\geq \frac{2n}{3} \cdot w(A, B)
        \geq \frac{2n}{3} \cdot \Phi_{G} \cdot \min\{\vol(A), \vol(B)\}\\
        &\geq \frac{2n}{3} \cdot \Phi_G \cdot d_{\min} \cdot \frac{n}{3}\\
        &= \frac{2\Phi_G\cdot n^2 \cdot d_{\min}}{9}.
    \end{align*}
  Combining the two cases above gives us that 
    \[
        \mathrm{cost}_G(\T^*) \geq \frac{2\Phi_G}{9} \cdot \max \left\{\frac{\vol(G)^2}{d_{\max}} , d_{\min} \cdot n^2 \right\},
    \]
    which proves the statement. 
\end{proof}

\begin{proof}[Proof of Lemma~3.1]
    Let $\T$ be an arbitrary $\textsf{HC}$ tree, and $\T^*$ be an optimal tree. Combining Lemmas~\ref{lem:Trivial bound} and \ref{lem:Lower Bound Cost Expander},  we have that
    \begin{align*}
    \frac{\cost_G(\T)}{\cost_G(\T^*)} & \leq \frac{n\cdot\vol(G)}{2}\cdot \frac{9}{2\Phi_{G}\cdot n \cdot \vol(G)\cdot \max\left\{\frac{d_{\mathrm{avg}}}{d_{\max}}, \frac{d_{\min}}{d_{\mathrm{avg}}}\right\}} \\
    & = \frac{9}{4\cdot \Phi_{G}}\cdot \min\left\{\frac{d_{\mathrm{avg}}}{d_{\min}}, \frac{d_{\max}}{d_{\mathrm{avg}}} \right\}, 
    \end{align*}
    which proves the statement.
\end{proof}

\subsection{The algorithm for graphs of high conductance }\label{sec:Alg for expanders}
While Lemma~\ref{lem:simplebound} holds for  any graph $G$, it implies  some interesting facts for expander graphs: first of all, when $G=(V,E,w)$ satisfies $d_{\max}/d_{\min}=O(1)$ and $\Phi_G=\Omega(1)$, Lemma~\ref{lem:simplebound} shows that any \textsf{HC} tree $\T$ is an $O(1)$-approximate tree.  
In addition, although $\Phi_G$ plays a crucial role in analysing $\cost_G(\T)$ as for many other graph problems, Lemma~\ref{lem:simplebound} indicates that the degree distribution of $G$ might also have a significant impact. One could naturally ask the extend to which the degree distribution of $V(G)$ would influence the construction of $\OPT_G$.
To answer this question, we study the following example.

\begin{example}\label{ex:Expander Bad Example}
    We study the following graph $G$, in which  all the edges have unit weight:
    \begin{enumerate}
        \item Let $G_1=(V, E_1)$ be a constant-degree expander graph of $n$ vertices with $\Phi_{G_1}=\Omega(1)$, e.g., the ones presented in \cite{HooLinWig06}; 
        \item We choose $\lfloor n^{2/3} \rfloor$ vertices from $V$ to form $S$, and let $K=(S, S\times S)$ be a complete graph defined on $S$; 
        \item Partition the vertices of $V \setminus S$ into $\lfloor n^{2/3} \rfloor$ groups of roughly the same size, associate each group to a unique vertex in $S$, and let $E_2$ be the set of edges formed by connecting every vertex in $S$ with all the vertices in its associated group;
        \item We define $G\triangleq (V, E_1\cup (S\times S) \cup E_2)$, see Figure~\ref{Fig:Expande_Clique_Example}(a) for illustration.
    \end{enumerate}
    By construction, we know that $\Phi_G =\Omega(1)$, and the degrees of $G$ satisfy $d_{\max}=\Theta(n^{2/3})$, $d_{\min}=\Theta(1)$, and $d_{\mathrm{avg}}=\Theta(n^{1/3})$ as $\vol(G)=\Theta(n^{4/3})$. 
    Therefore, the ratio between $\cost_G(\T)$ for any \textsf{HC} tree $\T$ and $\OPT_G$ could be as high as $\Theta(n^{1/3})$. 
    On the other side, it is not difficult to show that the tree $\T^*$ illustrated in Figure~\ref{Fig:Expande_Clique_Example}(b), which first separates the set $S$ of high-degree vertices from $V\setminus S$ at the top of the tree, 
    actually $O(1)$-approximates $\OPT_G$. To see this, notice that $\cost(\T^*)~\leq~\cost(\mathcal{\T^*}[S]) + n \cdot\vol(G_1) + n^2 = \Theta(n^2)$, as the complete subgraph $G[S]$ induces a cost of $\Theta((n^{2/3})^3)$ \cite{dasgupta2016cost}. The existence of the subgraph $G[S]$ also implies that $\OPT_G = \Omega(n^2)$.
\end{example}
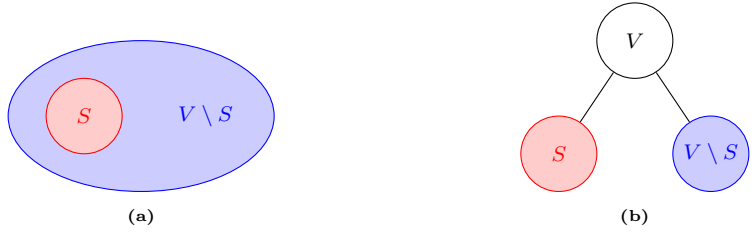
\begin{figure}[h]
    \centering
    \begin{subfigure}[t]{.4\textwidth}
        \centering
        \begin{tikzpicture}[scale = 0.5]
            \draw[color=blue, fill=blue!20] (0,0) circle [x radius=3.5cm, y radius=2cm] node[anchor=south east, right=10pt]{\scriptsize{$V \setminus S$}};
            \draw[color=red, fill=red!20] (-1.5, 0) circle [radius=1cm] node {\scriptsize{$S$}};
        \end{tikzpicture}
        \caption{}
    \end{subfigure}
    \begin{subfigure}[t]{.4\textwidth}
        \centering
        \begin{tikzpicture}[sibling distance=2cm]
            \node [draw, circle,inner sep=1pt, minimum size=1cm](first) {\scriptsize{$V$}}
            child { node [draw, circle, color = red,fill=red!20,inner sep=1pt, minimum size=1cm] {\scriptsize{$S$}} 
            }
            child { node [draw, circle, color = blue,fill=blue!20,inner sep=1pt, minimum size=1cm] {\scriptsize{$V\setminus S$}} 
            };
        \end{tikzpicture}
        \caption{}
    \end{subfigure}
    \caption{\small{(a) Our constructed graph $G$; (b) the tree that separates the vertices of high degrees from the others achieves $O(1)$-approximation.}}
    \label{Fig:Expande_Clique_Example}
\end{figure}
 

This example suggests that grouping vertices of similar degrees  first would potentially help reduce $\cost_G(\T)$ for our constructed $\T$. This motivates us to design the following   Algorithm~\ref{algo:degree} to construct an \textsf{HC} tree, and the algorithm's performance is summarised in Theorem~\ref{thm:degree}. 
We highlight that the output of  Algorithm~\ref{algo:degree} is uniquely  determined by the ordering of the vertices of $G$ according to their degrees, which can be computed in $O\left(n\cdot\log n\right)$ time.

\begin{algorithm}[H]
\DontPrintSemicolon
  \KwInput{$G=(V,E,w)$ with the ordered vertices such that  $d_{v_1}\geq\ldots\geq d_{v_{|V|}}$;}
  \KwOutput{An HC tree $\T_{\deg}(G)$;}
  \If{$|V|=1$}
    {
        \Return the single vertex in $V$ as the tree;    
    }
    \Else 
    {
    	$i_{\max}:=\lfloor \log_2(|V|-1) \rfloor$;
    	 $r := 2^{\imax}$;
        $A:=\left\{v_1,\ldots, v_{r}\right\}$;
        $B:= V\setminus A$; 
    
        Let $\T_1 := \mathrm{\texttt{HCwithDegrees}}(G\{A\}); 
        \T_2 := \mathrm{\texttt{HCwithDegrees}}(G\{B\})$;
        
      \Return $\T_{\deg}$ with $\T_1$ and $\T_2$ as the two children.
    }
\caption{\texttt{HCwithDegrees}$(G\{V\})$ \label{algo:degree}}
\end{algorithm}

\begin{theorem}\label{thm:degree}
Given any graph $G=(V,E,w)$ with  conductance $\Phi_G$ as input, Algorithm~\ref{algo:degree} runs in $O(m + n \log n)$ time, and returns an HC tree $\T_{\deg}$ of $G$ that satisfies 
$\mathrm{cost}_G(\T_{\deg}) = O\left(1/\Phi_G^4\right)\cdot\OPT_G$.
\end{theorem}

Theorem~\ref{thm:degree} shows that, when the input   $G$ satisfies  $\Phi_G=\Omega(1)$, the output $\T_{\deg}$ of Algorithm~\ref{algo:degree}  achieves an $O(1)$-approximation.
 It is important to notice that, while it is known \cite{charikar2017approximate} that  there is no polynomial-time algorithm  that $O(1)$-approximates   $\mathsf{OPT}_G$ for a general graph $G$ under the \textsf{SSEH}, our result shows that an $O(1)$-approximate \textsf{HC} tree can be constructed in polynomial-time for expander graphs. Our result  is in line with a sequence of research showing that this type of problems become easier when the input graphs exhibit a good expansion property~(e.g., \cite{ABS15,AKKSTV08,Kolla10,Li0Z19}). 
 To the best of our knowledge, Theorem~\ref{thm:degree} is the first such result for hierarchical clustering on graphs of high expansion conductance. 
Moreover, as the high-conductance property can be determined in nearly-linear time by computing $\lambda_2(\mathcal{L}_G)$ 
 and applying the Cheeger inequality, Algorithm~\ref{algo:degree} presents a very simple construction of an $O(1)$-approximate \textsf{HC} tree once the input $G$ is known to have high conductance.

\subsection{Analysis of the algorithm}\label{sec:Analysis expanders}
In this subsection we analyse Algorithm~\ref{algo:degree}, and   prove Theorem~\ref{thm:degree}. 

\subsubsection{The dense branch and its properties}
Our analysis is crucially based on the notion of \emph{dense branch}, which can be informally described as follows: for any given $\T$,
we perform a traversal in $\T$ starting at its root node $A_0$ and sequentially travel to the child of \emph{higher} volume. The process stops whenever we reach a node $A_k$, for some $k \in \mathbb{Z}_{\geq 0}$, such that $\vol(A_k) > \vol(G)/2$ and both of its children have volume at most $\vol(G)/2$. The sequence of  visited nodes in this process is the dense branch of $\T$. Formally, we define the dense branch as follows:

\begin{definition}[Dense branch]
    Given a graph $G$ with an \textsf{HC} tree $\T$, the \emph{dense branch} is the path $(A_0, A_1, \dots, A_k)$ in $\T$, for some $k \in \mathbb{Z}_{\geq 0}$, such that the following hold:
    \begin{enumerate}
        \item $A_0$ is the root of $\T$;
        \item $A_k$ is the node such that $\vol(A_k) > \vol(G)/2$ and both of its children have volume at most $\vol(G)/2$.
    \end{enumerate}
\end{definition}

It is important to note that the dense branch of $\T$ is unique, and consists of all the nodes $A_i$ with $\vol(A_i) > \vol(G)/2$. Moreover, for every pair of consecutive nodes $A_i, A_{i+1}$ on the dense branch, $A_{i+1}$ is the child of $A_i$ of the \emph{higher} volume. 
Now we will present some properties of the dense branch, which will be used extensively in our analysis.

\begin{lemma}[Lower bound of $\cost_G(\T)$ based on the dense branch]\label{lem:Cost light nodes}
    Let $G$ be a graph of conductance $\Phi_G$, and let $\T$ be an arbitrary \textsf{HC} tree of $G$. Suppose $(A_0, \dots, A_k)$ is the dense branch of $\T$, for some $k \in \mathbb{Z}_{\geq 0}$, and suppose each node $A_i$ has sibling $B_i$, for all $1 \leq i \leq k$. Then,
    the following lower bounds of 
    $\cost_G(\T)$ hold:
    \begin{enumerate}
        \item $\cost_G(\T) \geq \frac{\Phi_G}{2}\sum_{i=1}^k |A_{i-1}| \cdot \vol(B_i)$;
        \item $\cost_G(\T) \geq \frac{\Phi_G}{2} \cdot |A_k| \cdot \vol(A_k)$.
    \end{enumerate}
\end{lemma}

\begin{proof}
    Let $\mathcal{S} = \{A_k, B_1, \dots, B_k\}$. We   focus on the edges $e$ crossing different pairs $(X, Y)$ of $X, Y \in \mathcal{S}$, and  have that \begin{align*}
        \cost_G(\T) 
        &\geq \sum_{\substack{(X, Y) \in \mathcal{S}^2 \\ X \neq Y}} \sum_{e \in E(X, Y)} \cost_{\T}(e)
        = \frac{1}{2}\cdot \sum_{X \in \mathcal{S}} \sum_{e \in E(X, V \setminus X)} \cost_{\T}(e)\\
        &\geq \frac{1}{2} \cdot \sum_{i=1}^k \sum_{e \in E(B_i, V \setminus B_i)} \cost_{\T}(e)
        \geq \frac{1}{2} \cdot \sum_{i=1}^k \sum_{e \in E(B_i, V\setminus B_i)}
        w_e \cdot |A_{i-1}| \\
        &= \frac{1}{2} \cdot \sum_{i=1}^k |A_{i-1}| \cdot w(B_i, V\setminus B_i)
        \geq \frac{\Phi_G}{2} \cdot \sum_{i=1}^k |A_{i-1}| \cdot \vol(B_i),
    \end{align*}
        where the third inequality holds by the fact that  $A_{i-1}$ is the parent of $B_i$ and any edge $e$ with  exactly one endpoint in $B_i$ satisfies $\cost_{\T}(e) \geq w_e \cdot |A_{i-1}|$, and the last inequality holds by the fact that   $G$ has conductance $\Phi_G$ and $\vol(B_i) \leq \vol(G)/2$ by the definition of the dense branch.

    For the second bound, let $A_{k+1}, B_{k+1}$ be the children of $A_k$ such that $\vol(A_{k+1}) \geq \vol(B_{k+1})$. Therefore, we have that 
    \begin{align*}
        \cost_G(\T) 
        &\geq \sum_{e \in E(A_{k+1}, V \setminus A_{k+1})} \cost_T(e)
        \geq \sum_{e \in E(A_{k+1}, V \setminus A_{k+1})} w_e \cdot |A_k| \\
        &\geq |A_k| \cdot w(A_{k+1}, V \setminus A_{k+1})
        \geq |A_k| \cdot \Phi_G \cdot \vol(A_{k+1})\\
        &\geq \frac{\Phi_G}{2} \cdot |A_k| \cdot \vol(A_k).\qedhere
    \end{align*}
\end{proof}

Next, we prove that, if the dense branch of an optimal tree $\T^*$ consists of a single node, i.e., the root of $\T^*$, then any tree $\T$ achieves a $(1/\Phi_G)$-approximation.
 If we assume this corner case occurs, we are able to prove a stronger result over Theorem~\ref{thm:degree}, and this result is presented in the following lemma.

\begin{lemma}\label{lem:Short dense branch}
    Let $G$ be a graph of conductance $\Phi_G$, and let $\T^*$ be any optimal \textsf{HC} tree of $G$ such that its dense branch only consists of the root of $\T^*$. Then, it holds  for any \textsf{HC} tree $\T$ of $G$ that
    \[
        \cost_G(\T) \leq \frac{\cost_G(\T^*)}{\Phi_G}.
    \]
    In particular, it implies that $\cost_G(\T_{\deg}) \leq 1/\Phi_G \cdot \OPT_G$.
\end{lemma}
 
\begin{proof}
     Suppose the dense branch of $\T^*$ only consists of the root node $A_0$, and  let $A_1, B_1$ be the two children of $A_0$ such that $\vol(A_1) \geq \vol(B_1)$. Since $A_1$ does not belong to the dense branch, we know that $\vol(A_1) \leq \vol(G)/2$ and hence $\vol(A_1) = \vol(B_1) = \vol(G)/2$. Therefore, we have that 
    \[
        \cost_G(\T^*) 
        \geq \sum_{e \in E(A_1, B_1)} \cost_{\T^*}(e)
        = n \cdot w(A_1, B_1) 
        \geq n \cdot \Phi_G \cdot \frac{\vol(G)}{2},
    \]
    where the last inequality uses that $G$ has conductance $\Phi_G$ and $\vol(A_1) = \vol(B_1) = \vol(G)/2$. On the other hand,   by  Lemma~\ref{lem:Trivial bound}  it holds for any tree $\T$ that
    \[
        \cost_G(\T) \leq n \cdot \frac{\vol(G)}{2}.
    \]
    Combining the two inequalities above proves the claimed statement.
\end{proof}

 Finally, we prove that if the size of the last node on the dense branch of an optimal tree $\T^*$ is significantly large, then any tree $\T$ achieves an $O(1/\Phi_G)$-approximation. Again, if we assume this corner case occurs, we prove a stronger result over Theorem~\ref{thm:degree} and we present the result in the following lemma.
\begin{lemma}\label{lem:Only large nodes on dense branch}
    Let $G$ be a graph of conductance $\Phi_G$, and let $\T^*$ be any optimal \textsf{HC} tree of $G$ whose dense branch is $(A_0, \dots, A_k)$, for some $k \in \mathbb{Z}_{\geq 0}$. If $|A_k| \geq (n-1) / 2$, then it holds for any \textsf{HC} tree $\T$ of $G$ that 
    \[
        \cost_G(\T) \leq \frac{8\cdot \cost_G(\T^*)}{\Phi_G}.
    \]
    In particular, it implies that $\cost_G(\T_{\deg}) \leq 8/\Phi_G \cdot \OPT_G$.
\end{lemma}
\begin{proof}
    We apply the second property of  Lemma~\ref{lem:Cost light nodes} for the tree $\T^*$, and have that
     \[
        \cost_G(\T^*) 
        \geq \frac{\Phi_G}{2} \cdot |A_k| \cdot \vol(A_k)
        \geq \frac{\Phi_G}{2} \cdot \frac{n-1}{2} \cdot \frac{\vol(G)}{2}
        \geq \frac{\Phi_G}{8} \cdot \frac{n \cdot \vol(G)}{2},
     \]
    where the second inequality uses that $\vol(A_k) > \vol(G)/2$ and our assumption on the size of $A_k$. On the other hand, by Lemma~\ref{lem:Trivial bound} it holds for any tree $\T$ that
    \[
        \cost_G(\T) \leq n \cdot \frac{\vol(G)}{2}.
    \]
    Combining the two inequalities above proves the claimed statement.
\end{proof}

\subsubsection{Proof of Theorem~\ref{thm:degree}}
 


 Now we  prove the main result of this section, i.e., Theorem~\ref{thm:degree}. To sketch the main proof ideas,  we start with  an optimal \textsf{HC} tree $\T_0$ of $G$,  and construct trees $\T_1, \T_2, \T_3$ and $\T_4$ with the following properties: 
\begin{enumerate}
    \item  $\cost_G(\T_i) $ can be upper bounded with respect to $\cost_G\left(\T_{i-1}\right)$ for every $1\leq i\leq 4$;
    \item The final constructed tree $\T_4$ is exactly the tree $\T_{\deg}(G)$, the output of Algorithm~\ref{algo:degree}. 
\end{enumerate}
Combining these two facts allows us to upper bound $\cost_G(\T_{\deg})$ with respect to $\OPT_G$. In the remaining part of the section, we assume that the dense branch of the optimal tree $\T_0$ contains at least two nodes, since otherwise  one can apply Lemma~\ref{lem:Short dense branch} to obtain a stronger result than Theorem~\ref{thm:degree}.  Moreover, if $A_{k_0}$ is the last node on the dense of $\T_0$ we assume that $|A_{k_0}| < (n-1)/2$, as otherwise by Lemma~\ref{lem:Only large nodes on dense branch} we again obtain a stronger result than Theorem~\ref{thm:degree}.

\paragraph{Step~1: Regularisation.}

Let $(A_0, \dots, A_{k_0})$, for some $k_0 \in \mathbb{Z}_{+}$, be the dense branch of $\T_0$, and let $B_i$ be the sibling of $A_i$, for all $1\leq i\leq k_0$. Let $\imin = \lfloor \log_2 |A_{k_0}| \rfloor$,  $i_{\max}=\lfloor \log_2(|V|-1) \rfloor$,  and notice that  $\imin < \imax$\footnote{This follows from the assumption that $|A_{k_0}| < (|A_0| - 1) / 2$.}. 
The goal of this step is to adjust the dense branch of $\T_0$, such that the resulting tree $\T_1$ satisfies the following conditions:
\begin{enumerate}
    \item $|A_1| \geq 2^{\imax}$;
    \item For all  $i\in (\imin, \imax]$, there is a node of size exactly $2^i$ along the dense branch of $\T_1$;
    \item  If $|A_{k_0}| > 2^{\imin}$, then $A_{k_0}$ has a child of size $2^{\imin}$. 
\end{enumerate}

 We deal with each condition individually, and we start with the first one. If the first condition is not already satisfied,  we perform an initial adjustment to the tree $\T_0$ to ensure that, in the resulting tree $\T_0'$ the second node on its dense branch has size $2^{\imax}$, as illustrated in Figure~\ref{fig:Adjustment 1}(a). 
 Specifically, we consider an \emph{arbitrary} partition of $B_1=B_1^1 \cup B_1^2$ such that $|A_1| + |B_1^1| = 2^{i_{\max}}$. We adjust $\T_0$ as follows: we set the two children of $A_0$ in $\T_0'$ as some newly created node $A'_1$ and $B_1^2$, and set the two children of $A_1'$ as $A_1$ and $B_1^1$; the remaining part of $\T_0'$ is the same as $\T_0$. As such, the dense branch of the new tree $\T_0'$   becomes $(A_0, A_1', A_1,\ldots, A_{k_0})$, and $|A_1'|=2^{i_{\max}}$ by construction.

 To obtain the tree satisfying the second condition in the mean time,  we apply a sequence of  adjustments, each of which creates a new node of exact size $2^i$ for some suitable $i$.
Specifically, let  $i\in(\imin,\imax]$  be the largest integer  such that there is no node of size $2^i$ on the dense branch of $\T_0'$.
Since $|A_{k_0}| \geq 2^{\imin}$, there is some  node $A_j$ on the dense branch such that $|A_j| > 2^i$ and $|A_{j+1}| < 2^i$. 
We adjust the branch at $A_j$ as follows: 
(i) we consider a partition of $B_{j+1} = B_{j+1}^1 \cup B_{j+1}^2$ such that $|A_{j+1}| + |B_{j+1}^2| = 2^i$;
(ii) we replace the node $A_j$ by some newly created node $A_j^1$ that has children $B_{j+1}^1$ and a new node $A_j^2$; 
(iii) the two children of $A_j^2$ will be $A_{j+1}$ and $B_{j+1}^2$. This adjustment is illustrated in Figure~\ref{fig:Adjustment 1}(b), and we repeat this process until no such $i$ exists anymore. 

 To ensure that the third condition is satisfied (assuming $|A_{k_0}| > 2^{\imin}$), we adjust the dense branch at $A_{k_0}$ in a similar way as before. Specifically, let $A_{k_0 + 1}, B_{k_0 + 1}$ be the two children of $A_{k_0}$ such that $A_{k_0 + 1}$ is the child of smaller size $|A_{k_0 + 1}| < 2^{\imin}.$ We consider an arbitrary partition of $B_{k_0 + 1} = B_{k_0 + 1}^1 \cup B_{k_0 + 1}^2$ such that $|A_{k_0 + 1}| + |B_{k_0 + 1}^2| = 2^{\imin}$. We replace the node $A_{k_0}$ with a new node $A_{k_0}^1$ that has children $B_{k_0 + 1}^1$ and a new node $A_{k_0}^2$. The two children of $A_{k_0}^2$ are $A_{k_0 + 1}$ and $B_{k_0 + 1}^2$. After this transformation, one of the following two conditions happens: 
(i) the dense branch increases its size by one, having the final node $A_{k_0}^2$ of size $|A_{k_0}^2| = 2^{\imin}$, or 
(ii) $A_{k_0}^1$ is the final node on the dense branch having size $|A_{k_0}^1| = |A_{k_0}| > 2^{\imin}$ and $A_{k_0}^1$ has one child $A_{k_0}^2$ of size $|A_{k_0}^2| = 2^{\imin}$. In both cases, the third condition is satisfied.
 
We call the resulting tree $\T_1$, and the following lemma gives an upper bound of $\cost_G(\T_1)$.

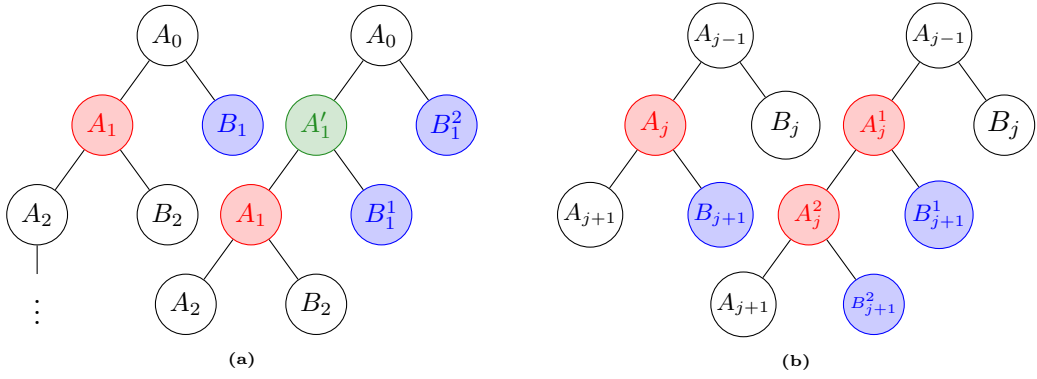
\begin{figure}[ht]
  \centering
  \begin{subfigure}[h]{0.45\textwidth}
    \centering
    \begin{tikzpicture}[scale=0.66, level distance = 1.8cm, sibling distance=2.6cm]
      \node [draw, circle,inner sep=1pt, minimum size=0.8cm](first) {\small{$A_{0}$}}
        child { node [draw, circle, color = red,fill=red!20,inner sep=1pt, minimum size=0.8cm] {\small{$A_{1}$}}
            child { node [draw, circle,inner sep=1pt, minimum size=0.8cm] {\small{$A_{2}$}} 
                child{ node { \vdots } }    
            }
            child { node [draw, circle,inner sep=1pt, minimum size=0.8cm] {\small{$B_{2}$}} }
        }
        child { node [draw, circle, color = blue,,fill=blue!20,inner sep=1pt, minimum size=0.8cm] {\small{$B_1$}}};
        \node[draw, circle, right=of first, right = 2cm,inner sep=1pt, minimum size=0.8cm]{ \small{$A_{0}$}}
            child { node [draw, circle, color = forestgreen, fill=forestgreen!20,inner sep=1pt, minimum size=0.8cm] {\small{$A_1'$}} 
                child { node [draw, circle, color = red, fill=red!20,inner sep=1pt, minimum size=0.8cm] {\small{$A_1$} } 
                    child { node [draw, circle,inner sep=1pt, minimum size=0.8cm] {\small{$A_{2}$}} }
                    child { node [draw, circle,inner sep=1pt, minimum size=0.8cm]{\small{$B_2$}} }
                }
                child { node [draw, circle, color = blue,fill=blue!20,inner sep=1pt, minimum size=0.8cm] {\small{$B_{1}^1$}} }
            }         
            child{ node [draw, circle, color = blue,fill=blue!20,inner sep=1pt, minimum size=0.8cm] {\small{$B_1^2$}}}; 
    \end{tikzpicture}
    \caption{}
  \end{subfigure}
  \begin{subfigure}[h]{0.45\textwidth}
  \centering
        \begin{tikzpicture}[nodes ={draw, circle}, scale=0.66, level distance = 1.8cm, sibling distance=2.6cm]
            \node [inner sep=1pt, minimum size=0.8cm] (first) {\footnotesize{$A_{j-1}$}}
            child { node [color = red, fill=red!20,inner sep=1pt, minimum size=0.8cm] {\small{$A_{j }$}} 
                child { node [inner sep=1pt, minimum size=0.8cm] {\footnotesize{$A_{j+1}$}} }
                child { node [color = blue, fill=blue!20,inner sep=1pt, minimum size=0.8cm] {\footnotesize{$B_{j+1}$}} }
            }
            child { node [minimum size=0.8cm] {\small{$B_j$}} };
        
            \node[right=of first, right = 2cm,inner sep=1pt, minimum size=0.8cm]{ \footnotesize{$A_{j-1}$}}
                child { node [color = red, fill=red!20,inner sep=1pt, minimum size=0.8cm] {\footnotesize{$A_j^1$}} 
                    child { node [color = red,fill=red!20,inner sep=1pt, minimum size=0.8cm] {\footnotesize{$A_j^2$} } 
                        child { node [inner sep=1pt, minimum size=0.8cm] {\footnotesize{$A_{j+1}$}} }
                        child { node [color = blue, fill=blue!20,inner sep=1pt, minimum size=0.8cm]{\tiny{$B_{j+1}^2$}} }
                    }
                    child { node [color = blue,fill=blue!20,inner sep=1pt, minimum size=0.8cm] {\footnotesize{$B_{j+1}^1$}} }
                }         
                child{ node [inner sep=1pt, minimum size=8pt, minimum size=0.8cm] {$B_j$}}; 
        \end{tikzpicture} 
        \caption{}
  \end{subfigure}
  \caption{\small{(a) With a proper partition of $B_1$, we have a new node $A_1'$ in $\T_0'$ such that   $|A_1'|=2^{i_{\max}}$; (b)~With a proper partition of $B_{j+1}$, we have a new node $A_j^2$ of size $2^i$.}   \label{fig:Adjustment 1} }

\end{figure}

\begin{lemma}\label{lem:Step 1}
Our constructed tree   $\T_1$ satisfies $
        \cost_G(\T_1) \leq \lp 1 +  4/\Phi_G \rp \cdot \cost_G(\T_0)$.
\end{lemma}

\begin{proof}
 Let $F$ be the set of edges whose cost increases due to our  adjustment to $\T_0$, i.e., $ F = \{ e : \cost_{\T_0}(e) \leq \cost_{\T_1}(e) \}$. Then, we know that 
    \begin{align}
        \cost(\T_1) 
        &= \sum_{e\in E} \cost_{\T_1}(e) 
        = \sum_{e \in F} \cost_{\T_1}(e) + \sum_{e \in E \setminus F} \cost_{\T_1}(e) \nonumber \\
        &\leq \sum_{e \in F} \cost_{\T_1}(e) + \sum_{e \in E \setminus F} \cost_{\T_0}(e)
        \leq  \lp \sum_{e \in F} \cost_{\T_1}(e)\rp + \cost(\T_0). \label{eq:step1}
    \end{align}
    Hence, it suffices to show that 
    \[
        \sum_{e \in F} \cost_{\T_1}(e) \leq \frac{4}{\Phi_G} \cdot \cost(\T_0),
    \]
    and we sometimes refer to $\sum_{e \in F} \cost_{\T_1}(e)$ as the \emph{additional cost} of the transformation from $\T_0$ to $\T_1$.
    We look at  the additional cost introduced by one of our adjustments, say at node $A_j$, as illustrated in Figure~\ref{fig:Adjustment 1}(b). 
    It is not difficult to show  that the cost can increase only for edges $e \in E(B_{j+1}^1, B_{j+1}^2)$, for each of which the additional cost is  exactly $\cost_{\T_1}(e) = w_e \cdot |A_{j}^1| = w_e \cdot |A_j|$. Hence,  the total additional cost of the adjustment for $A_j$ is at most $|A_j| \cdot w(B_{j+1}^1, B_{j+1}^2)$, and after all adjustments, including the initial adjustment to transform $\T_0$ in $\T_0'$    and the final adjustment at node $A_{k_0}$, the total additional cost is upper bounded by  
    \begin{align}
        \sum_{e \in F} \cost_{\T_1}(e) 
        &\leq \sum_{i=1}^{k_0 + 1} \sum_{e \in E(B_i, B_i)} w_e \cdot |A_{i-1}| 
         \leq \sum_{i=1}^{k_0 + 1}  |A_{i-1}| \cdot \vol(B_i) \nonumber\\
        &\leq  \lp \sum_{i=1}^{k_0}  |A_{i-1}| \cdot \vol(B_i) \rp   + |A_{k_0}| \cdot \vol(A_{k_0})
        \leq \frac{4}{\Phi_G} \cdot \cost(\T_0),\label{eq:step2}
    \end{align}
    where the last inequality follows by Lemma~\ref{lem:Cost light nodes}. Combining (\ref{eq:step1}) with (\ref{eq:step2}) proves the statement.
\end{proof}


\paragraph{Step~2: Compression.}
With potential relabelling of the nodes, let $(A_0, \dots, A_{k_1})$ be the dense branch of $\T_1$, for some $k_1 \in \mathbb{Z}_{+}$ satisfying $|A_{k_1}| \geq 2^{\imin}$, and  $B_i$ be the sibling of $A_i$. The objective of this step is to ensure that, by a sequence of adjustments, all the nodes along the dense branch are of size equal to some power of $2$. In this step we perform two similar types of adjustments: the first type concerns all nodes $A_j$ of size at least $2^{\imin + 1}$, and the second type concerns the nodes $A_j$ of size $2^{\imin} \leq |A_j| < 2^{\imin + 1}$.
We  begin by describing how an adjustment of the first type is performed, and refer the reader to Figure~\ref{fig:Adjustment 2} for illustration.
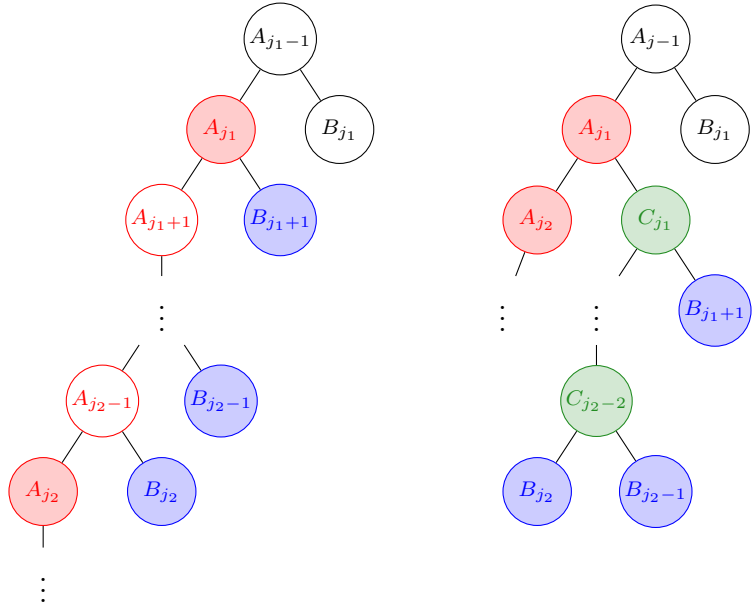
\begin{figure}
    \centering
    \begin{tikzpicture}[scale=0.6, level distance = 2cm, sibling distance=2.6cm,minimum size=0.9cm,inner sep=1pt]
          \node [draw, circle, inner sep=1pt ](first) { \scriptsize{$A_{{j_1}-1}$} }
            child { node [draw, circle, color = red, fill=red!20,inner sep=1pt ] {\scriptsize{$A_{j_1}$}} 
                child { node [draw,  circle, color = red,inner sep=1pt ] {\scriptsize{$A_{j_1+1}$}} 
                    child { node { \vdots } 
                        child { node [draw, circle, color = red,inner sep=1pt] {\scriptsize{$A_{j_2 -1}$}} 
                            child { node [draw, circle, color = red, fill=red!20, inner sep=1pt] {\scriptsize{$A_{j_2}$}}
                                child { node {\vdots} }
                            }
                            child { node [draw, circle, color = blue, fill=blue!20, inner sep=1pt] {\scriptsize{$B_{j_2}$}} }
                        }
                        child { node [draw, circle, color = blue,color = blue, fill=blue!20, inner sep=1pt] {\scriptsize{$B_{j_2 - 1}$}} }
                    }
                }
                child { node [draw, circle, color = blue,color = blue, fill=blue!20, inner sep=1pt] {\scriptsize{$B_{j_1+1}$} }}
            }
            child { node [draw, circle,inner sep=1pt] {\scriptsize{$B_{j_1}$}} };
            
            \node[right=of first, draw, circle, right = 4cm,inner sep=1pt]{ \scriptsize{$A_{j-1}$}}
                child { node [draw, circle, color = red, fill=red!20,inner sep=1pt] {\scriptsize{$A_{j_1}$}} 
                    child { node [draw, circle, color = red, fill=red!20,inner sep=1pt] {\scriptsize{$A_{j_2}$}}
                        child[left = 0.1cm,inner sep=1pt] {node { \vdots}}
                    }
                    child { node [draw, circle, color = forestgreen,fill=forestgreen!20, inner sep=1pt] {\scriptsize{$C_{j_1}$}}
                        child{ node {\vdots}
                            child { node [draw, circle, color = forestgreen,fill=forestgreen!20,inner sep=1pt]{\scriptsize{$C_{j_2 - 2}$}} 
                                child{node [draw, circle, color = blue,color = blue, fill=blue!20, inner sep=1pt] {\scriptsize{$B_{j_2}$ }}}
                                child{node [draw, circle,color = blue, fill=blue!20,   inner sep=1pt] {\scriptsize{$B_{j_2 - 1}$} }}
                            }
                        }
                        child { node [draw, circle, color = blue,color = blue, fill=blue!20, inner sep=1pt] {\scriptsize{$B_{j_1 +1}$}}}
                    }
                }
                child{ node [draw, circle,inner sep=1pt] 
                {\scriptsize{$B_{j_1}$}}}; 
    \end{tikzpicture}
    \caption{\small{The nodes between $A_{j_1}$ and $A_{j_2}$ (left) have size in $(2^{i-1}, 2^i)$; These nodes are compressed (right) and only the nodes of size some power of $2$ remain in the dense branch.}
    \label{fig:Adjustment 2}}
\end{figure}

Let $i\in(i_{\min} + 1, i_{\max}]$ be some index such that $|A_{j_1}|=2^i$ and  $|A_{j_2}|=2^{i-1}$ for some $j_1<j_2$. 
We compress the dense branch by removing all nodes between $A_{j_1}$ and $A_{j_2}$ as follows:
the two children of $A_{j_1}$ will be $A_{j_2}$ and some new node $C_{j_1}$, which has children $C_{j_1+1}$ and $B_{j_1+1}$; 
the two children of $C_{j_1+1}$ will be some new node $C_{j_1+2}$ and $B_{j_1+2}$, etc.
The last node $C_{j_2-2}$ has children $B_{j_2-1}$ and $B_{j_2}$.
In addition, we perform one more such adjustment to remove all nodes $A_j$ of size $2^{\imax} < |A_j| < n$.

 For the second type of adjustment, we have a case distinction based on the size of the final node $A_{k_1}$. Specifically, if $|A_{k_1}| = 2^{\imin}$, then we perform once more the previous adjustment for the value $i = \imin + 1$. On the other side, if $|A_{k_1}| > 2^{\imin}$, we perform the adjustment as illustrated in Figure~\ref{fig:Adjustment 2_2}: let $A_j$ be the node on the dense branch of size $2^{\imin + 1}$, which exists since $\imin + 1 \leq \imax$. Also, let $A_{k_1 + 1}, B_{k_1 + 1}$ be the two children of $A_{k_1}$ such that $|A_{k_1 + 1}| = 2^{\imin}$. Similar to before, we compress the dense branch by removing all nodes between $A_j$ and $A_{k_1 + 1}$ as follows: the two children of $A_j$ will be $A_{k_1 + 1}$ and a new node $C_j$ which has children $B_{k_1 + 1}$ and a new node $C_{j+1}$; the two children of $C_{j+1}$ are $B_{j + 1}$ and a new node $C_{j+2}$, etc. The last node $C_{k_1 - 1}$ has children $B_{k_1 - 1}$ and $B_{k_1}$. Notice that, after this adjustment, the final node on the dense branch is either $C_j$ if $\vol(C_j) > \vol(G)/2$, or $A_j$ otherwise. Since $|C_j| = 2^{\imin}$ and $|A_j| = 2^{\imin + 1}$, this ensures that all nodes (potentially except for $A_0$) on the dense branch have size $2^i$ for some $i \in [\imin, i_{\max}]$. We call the resulting tree $\T_2$, and the following lemma gives an upper bound of $\cost_G(\T_2)$.
\begin{figure}
    \centering
    \begin{tikzpicture}[scale=0.6, minimum size = 0.9cm, level distance = 2cm, sibling distance=2.6cm,inner sep=1pt]
          \node [draw, circle, inner sep=1pt ](first) { \scriptsize{$A_{j - 1}$} }
            child { node [draw, circle, color = red, fill=red!20, inner sep=1pt ] {\scriptsize{$A_{j}$}} 
                child { node [draw,  circle, color = red,inner sep=1pt ] {\scriptsize{$A_{j+1}$}} 
                    child { node { \vdots } 
                        child { node [draw, circle, color = red, inner sep=1pt] {\scriptsize{$A_{k_1}$}} 
                            child { node [draw, circle, color = red, fill=red!20, inner sep=1pt] {\scriptsize{$A_{k_1 + 1}$}}
                                child { node {\vdots} }
                            }
                            child { node [draw, circle, color = blue, fill=blue!20, inner sep=1pt] {\scriptsize{$B_{k_1 + 1}$}} }
                        }
                        child { node [draw, circle, color = blue,color = blue, fill=blue!20, inner sep=1pt] {\scriptsize{$B_{k_1}$}} }
                    }
                }
                child { node [draw, circle, color = blue,color = blue, fill=blue!20, inner sep=1pt] {\scriptsize{$B_{j+1}$} }}
            }
            child { node [draw, circle, inner sep=1pt] {\scriptsize{$B_{j}$}} };
            
            \node[right=of first, draw, circle, right = 4cm,inner sep=1pt]{ \scriptsize{$A_{j - 1}$}}
                child { node [draw, circle, color = red, fill=red!20,inner sep=1pt] {\scriptsize{$A_{j}$}} 
                    child { node [draw, circle, color = red, fill=red!20,inner sep=1pt] {\scriptsize{$A_{k_1+1}$}}
                        child[left = 0.1cm,inner sep=1pt] {node { \vdots}}
                    }
                    child { node [draw, circle, color = forestgreen,fill=forestgreen!20, inner sep=1pt] {\scriptsize{$C_{j}$}}
                        child { node [draw, circle, color = forestgreen, fill=forestgreen!20, inner sep=1pt]{\scriptsize{$C_{j + 1}$}}
                            child{ node {\vdots}
                                child { node [draw, circle, color = forestgreen,fill=forestgreen!20,inner sep=1pt]{\scriptsize{$C_{k_1 - 1}$}} 
                                    child{node [draw, circle, color = blue,color = blue, fill=blue!20, inner sep=1pt] {\scriptsize{$B_{k_1}$ }}}
                                    child{node [draw, circle,color = blue, fill=blue!20,   inner sep=1pt] {\scriptsize{$B_{k_1 - 1}$} }}
                                }
                            }
                            child { node [draw, circle, color = blue,color = blue, fill=blue!20, inner sep=1pt] {\scriptsize{$B_{j + 1}$}}}
                        }
                        child { node [draw, circle, color = blue,color = blue, fill=blue!20, inner sep=1pt] {\scriptsize{$B_{k_1 +1}$}}}
                    }
                }
                child{ node [draw, circle,inner sep=1pt] 
                {\scriptsize{$B_{j}$}}}; 
    \end{tikzpicture}
    \caption{\small{The nodes between $A_{j}$ and $A_{k_1 + 1}$ (left) have size in $(2^{\imin}, 2^{\imin + 1})$;  These nodes are compressed (right) and only the nodes of size some power of $2$ remain in the dense branch. }
    \label{fig:Adjustment 2_2}}
\end{figure}

\begin{lemma}\label{lem:Step 2}
    Our constructed tree $\T_2$ satisfies that  $\cost_G(\T_2) \leq 2 \cdot \cost_G(\T_1)$. 
\end{lemma}

\begin{proof}
We claim that the cost of every edge   at most doubles in $\T_2$, and this implies  that  $\cost(\T_2) \leq 2\cdot \cost(\T_1)$. To see why our claim holds, let us consider the set $$\mathcal{S} = \{B_1, B_2, \dots, B_{k_1}, B_{k_1 + 1}, A_{k_1 + 1}\},$$ 
    which consists of  the sibling of every node on the dense branch of $\T_1$, along with the children of $A_{k_1}$. This forms a partition of the vertex set. Now consider any two different internal nodes $X, Y \in \mathcal{S}$ such that $|\parent_{\T_1}(X)| \geq |\parent_{\T_1}(Y)|$, and an arbitrary edge $e \in E$. If $e \in E(X, X)$, then by construction we have  $\cost_{\T_1}(e) = \cost_{\T_2}(e)$. If $e \in E(X, Y)$, then $\cost_{\T_1}(e) = w_e \cdot |\parent_{\T_1}(X)|$. On the other hand, we have that $\cost_{\T_2}(e) \leq w_e \cdot |A_X|$, where $A_X$ is the first node on the path from $X$ to the root that appears on the dense branch of $\T_2$. By construction, we know that, if $2^{i -1} < |\parent_{\T_1}(X)| \leq 2^i$ for some $i$, then $|A_X| = 2^i \leq 2 \cdot \left| \parent_{\T_1}(X)\right|$. Therefore, we have that $\cost_{\T_2}(e) \leq 2 \cdot \cost_{\T_1}(e)$.
\end{proof}

\paragraph{Step~3: Matching.}
Let $(A_0, \dots, A_{k_2})$ be the dense branch of $\T_2$, for some $k_2\in\mathbb{Z}_+$, and  $B_i$ be the sibling of $A_i$. In this step we transform $\T_2$ into $\T_3$, such that $\T_3$ is isomorphic to $\T_{\deg}(G)$, which ensures that $\T_3$ and $\T_{\deg}(G)$ have the same structure.
To achieve this,  for every $1 \leq i \leq k_2$ we simply replace each $\T_2[B_i]$ with $\T_{\deg}\left(G\{B_i\}\right)$. We further replace $\T_2[A_{k_2}]$ with $\T_{\deg}\left(G\{A_{k_2}\}\right)$. We call the resulting tree $\T_3$, and  bound its cost by the following lemma:

\begin{lemma}\label{lem:Step 3}
Our constructed tree $\T_3$ satisfies that 
$
        \cost(\T_3) \leq \lp 1 + 4/\Phi_G \rp \cost(\T_2).$
\end{lemma}
\begin{proof}
Let $\mathcal{S} = \{ B_1, \dots, B_{k_2}, A_{k_2}\}$ be the set of internal nodes of $\T_2$  that undergo the above transformation. By construction, for every \emph{distinct} $X, Y \in \mathcal{S}$ and any  edge  $e \in E(X, Y)$, it holds that $\cost_{\T_2}(e) = \cost_{\T_3}(e)$. Therefore, the additional cost will be introduced only by  edges $e \in E(X,X)$, for some $X \in \mathcal{S}$. However, by
Lemma~\ref{lem:Trivial bound}
  we have that 
\[
    \sum_{X \in \mathcal{S}} \sum_{e \in E(X,X)} \cost_{\T_3}(e) 
    \leq \sum_{X \in \mathcal{S}} |X| \cdot \vol(X)
    =  \lp \sum_{i = 1}^{k_2} |B_i| \cdot \vol(B_i) \rp   + \left|A_{k_2}\right|\cdot  \vol(A_{k_2})
    \leq \frac{4}{\Phi_G} \cdot \cost(\T_2),
\]
where the last inequality follows by  Lemma~\ref{lem:Cost light nodes}. Combining the two cases together proves the  statement.
\end{proof}

\paragraph{Step~4: Sorting.}
We assume that $(A_0, \ldots, A_{k_3})$ is the dense branch of $\T_3$ for some $k_3\in\mathbb{Z}_+$, and we extend the dense branch to 
 $(A_0,\ldots, A_{i_{\max}})$
 with the   property that, for every $i\in[k_3,i_{\max}]$,  $A_i$ is the child of $A_{i-1}$ with the higher volume, and let 
$\mathcal{S}\triangleq \{B_1,\ldots, B_{i_{\max}}, A_{i_{\max}}\}$.  
Recall that in $\T_{\deg}(G)$ the  first $r = 2^{i_{\max}}$ vertices of the highest degrees, i.e., $\{v_1, \dots, v_{r}\}$ belong to $A_1$, of which the first $2^{i_{\max}-1}$ vertices 
belong to $A_2$, and so on; however, this might not be the case for $\T_3$. 
Hence, in the final step, we prove that $\T_3$ can be transformed into $\T_{\deg}(G)$ without a significant increase of the total cost. 
 In this step we will swap vertices between  the internal nodes $B_1, B_2, \dots, B_{i_{\max}}, A_{i_{\max}}$ in such a way that $A_{i_{\max}}$ will consist of $v_1$ and  $v_2$, $B_{i_{\max}}$ will consist of $v_3$ and  $v_4$, $B_{\imax - 1}$ will consist of $v_5$ up to $v_8$, etc.
We call  vertex $u$ \emph{misplaced} if the position of $u$ is $\T_3$ is different from the one in $\T_{\deg}(G)$. 
To transform $\T_3$ into $\T_{\deg}(G)$ we perform  a sequence of operations, each of which  consists in a chain of swaps focusing on the vertex of the \emph{highest} degree that is currently misplaced.
For the sake of argument, we assume that $v_1\not\in A_{i_{\max}}$ is misplaced, and we apply the following operation to move $v_1$ to $A_{i_{\max}}$: 
\begin{enumerate}
\item   Let $v_1\in B_{i_0}$ for some $i_0\geq 1$, and let $y$ be the vertex of the lowest degree among the vertices in $\mathcal{S}\setminus\{ B_1,\ldots, B_{i_0}\}$. 
 Say $y\in B_{i_1}$ for some $i_1> i_0$, and swap $v_1$ with $y$ while keeping the structure of the tree unchanged.
\item Repeat the swap operation above until $v_1$ reaches its correct place in $A_{i_{\max}}$.  
\end{enumerate}
Once the above process is complete and $v_1$   reaches $A_{\imax}$, we apply a similar  chain of swaps for $v_2$ to ensure $v_2$ also reaches $A_{\imax}$. Then, we sequentially apply the process for $v_3$ and $v_4$ to ensure they reach $B_{\imax}$, and continue this process until there are no more misplaced vertices.

We call the resulting tree $\T_3^{'}$, and notice that every node $X\in \mathcal{S}$ in  $\T_3^{'}$ contains the correct set of vertices. However, the positions of these vertices in $\T_3^{'}[X]$ might be different from the ones in $\T_{\deg}(G\{X\})$. To overcome this issue, we repeat Step~3 again to the tree $\T_3^{'}$, and this will introduce another factor of $(1+ 4/\Phi_G)$   to the total cost of the constructed tree. Importantly, the final constructed tree after this step is exactly the tree $\T_{\deg}(G)$ and we bound its cost by the following lemma:  

\begin{lemma}\label{lem:Step 4}
It holds for     $\T_{\deg} $  that  $
        \cost(\T_{\deg}) \leq \lp 1 + 24/\Phi_G \rp \lp 1 + 4/\Phi_G\rp\cdot  \cost(\T_3)$.
\end{lemma}
\begin{proof}
   We first prove that
    \[
        \cost\left(\T'_{3}\right) \leq \lp 1 + \frac{24}{\Phi_G}\rp \cdot \cost(\T_3).
    \]
    For the sake of analysis, we assume that $v_1\in B_{i_0}$ is misplaced, and let $(y_1,\ldots, y_t)$, for some $t\in\mathbb{Z}^+$, be the sequence of vertices with which $v_1$ performs the swap operations in order to reach $A_{i_{\max}}$. We first upper bound the additional cost that the swap between $y_1$ and $v_1$ introduces. Without loss of generality, we assume that $y_1\in B_{i_1}$ such that $|\parent_{\T_3}(B_{i_0})| > |\parent_{\T_3}(B_{i_1})|$. Let $e\in E(G)$ be any edge. Our analysis is based on the following case distinction:
    \begin{itemize}
    \item If $e \cap \{y_1,v_1\}=\emptyset$, then the swap between $y_1$ and $v_1$ would not change the cost of $e$;
    \item If $e=\{y_1, v_1\}$, then the swap between $y_1$ and $v_1$ would not change the cost of $e$  either;
    \item If $e$ is adjacent to $y_1$ or $v_1$, the cost of $e$ would increase by at most $$w_e\cdot~\left|\parent_{\T_3}(B_{i_0})\right|~=~w_e\cdot \left| A_{i_0-1} \right|.$$ 
    \end{itemize}
    Hence, the total additional cost of the swap between $y_1$ and $v_1$ is at most 
    \[
        \left|A_{i_0 -1}\right|\cdot ( d_{v_1} + d_{y_1}) \leq 2\cdot |A_{i_0 - 1}|\cdot d_{v_1}, 
    \] 
    as $v_1$ is the misplaced vertex of largest degree.

  We   apply the same analysis for the swap between $v_1$ and $y_2\in B_{i_2}$ for some $i_2\geq i_1$, implying that the additional cost by the swap between $y_2$ and $v_1$ is at most $2\cdot |A_{i_1 - 1}|\cdot d_{v_1}$, and so on. Therefore, the total additional cost introduced in order for $v_1$ to reach $A_{i_{\max}}$ is at most
    \[
    2\cdot d_{v_1}\cdot \left( |A_{i_0-1}| + \ldots + |A_{i_{t-1}-1}| \right) \leq 6\cdot d_{v_1}\cdot  |A_{i_0-1}|.
    \] 
    Since we only need to consider all the  misplaced vertices in some $B_j$, the total additional cost introduced over all sequences of swaps is at most
    \[
    6\cdot \sum_{i=1}^{i_{\max}}\vol(B_i) |A_{i-1}| \leq \frac{24}{\Phi_G}\cdot \cost(\T_3),
    \]
    where the last inequality follows by Lemma~\ref{lem:Cost light nodes}. In summary, we have that
    \[
    \cost\left(\T_3^{'}\right) \leq \left(1+ \frac{24}{\Phi_G}\right) \cost(\T_3).
    \]
    Since we use Step~3 in the end, we apply Lemma~\ref{lem:Step 3} once more and obtain   the claimed statement.
\end{proof}

\begin{proof}[Proof of Theorem~\ref{thm:degree}]
     By combining  Lemmas~\ref{lem:Step 1}, \ref{lem:Step 2}, \ref{lem:Step 3} and  \ref{lem:Step 4}, we prove the approximation guarantee of $\T_{\deg}$  in Theorem~\ref{thm:degree}. 
    The runtime of Algorithm~\ref{algo:degree} follows by a simple application of the master theorem.
\end{proof}  

\section{Hierarchical clustering for well-clustered graphs}\label{sec:Well-Clustered Graphs}

So far we   have shown that an $O(1)$-approximate \textsf{HC} tree can be easily constructed for expander graphs. In   this section we study a wider class of graphs that exhibit a clear structure of clusters, i.e., well-clustered graphs. 
Informally, a well-clustered graph is a collection of densely-connected components~(clusters) of high conductance, which are weakly interconnected. 
As these graphs form some of the most meaningful objects for clustering in practice, one would naturally ask whether our $O(1)$-approximation result for expanders can be extended to well-clustered graphs. In this section, we will give an affirmative answer to this question. 

To formalise the well-clustered property, we consider the notion of  \emph{$ (\Phi_{\mathrm{in}},\Phi_{\mathrm{out}})$-decomposition} introduced by Gharan and Trevisan~\cite{GT14}.  Formally, for a graph $G=(V,E, w)$ and $k\in\mathbb{Z}_{+}$, 
we say that $G$ has $k$ well-defined clusters if $V(G)$ can be partitioned into  disjoint subsets $\{P_i\}_{i=1}^k$, such that the following  hold:
\begin{enumerate}
\item There is a sparse cut between $S_i$ and $V\setminus P_i$ for any $1\leq i\leq k$, which is formulated as $\Phi_G(P_i)\leq\Phi_{\mathrm{out}}$;

\item Each induced subgraph $G[P_i]$ has high
conductance $\Phi_{G[P_i]}\geq \Phi_{\mathrm{in}}$.  
 
\end{enumerate}
We underline that, through the celebrated higher-order Cheeger inequality~\cite{higherCheeg}, this condition of $(\Phi_{\mathrm{in}},\Phi_{\mathrm{out}})$-decomposition can be approximately  reduced to other formulations of a well-clustered graph studied in the literature, e.g.,~\cite{PengSZ17, Luxburg07, ZhuLM13}.

 The rest of the section is structured as follows: we present the key notion  and the strong decomposition lemma~(Lemma~\ref{lem:Improved Decomposition}) used in our main algorithm in Section~\ref{sec:Partitioning well-clustered graphs and critical nodes}. 
In Section~\ref{sec:The algorithm for well-clustered graphs}, we present the algorithm for well-clustered graphs, whose performance is summarised in  Theorem~\ref{thm:main_k_clusters}. We analyse the algorithm's performance, and prove   Theorem~\ref{thm:main_k_clusters} in Section~\ref{sec:The analysis of the algorithm}. Lastly, 
we prove the  strong decomposition lemma in Section~\ref{sec:Proof of Improved Decomp Lemma}, due to the complexity of the underlying algorithm. 

\subsection{Partitioning well-clustered graphs and critical nodes}\label{sec:Partitioning well-clustered graphs and critical nodes}

The starting point of our second result is the following polynomial-time algorithm presented by  Gharan and Trevisan~\cite{GT14}, which produces a $(\Phi_{\mathrm{in}},\Phi_{\mathrm{out}})$-decomposition of  a graph $G$, for some parameters $\Phi_{\mathrm{in}}$ and $\Phi_{\mathrm{out}}$.
Specifically, given a well-clustered graph as input, their algorithm returns disjoint sets of vertices $\{P_i\}_{i=1}^{\ell}$ with bounded  $\Phi(P_i)$ and $\Phi_{G[P_i]}$ for each $P_i$, and the algorithm's performance is as follows:

\begin{lemma}[Theorem~1.6, \cite{GT14}]\label{lem:Gharan-Trevisan decomposition}
Let $G = (V, E, w)$ be a graph such that $\lambda_{k+1} > 0$, for some $k \geq 1$. Then, there is a  polynomial-time algorithm that finds an $\ell$-partition $\{P_i\}_{i=1}^{\ell}$ of $V$, for some $\ell \leq k$, such that the following hold for every $1 \leq i \leq \ell$: 
\begin{enumerate}[label=$(A\arabic*)$]\itemsep -2pt
    \item $\Phi(P_i) = O \lp k^6 \sqrt{\lambda_k} \rp$;
    \item $\Phi_{G[P_i]} = \Omega \lp \lambda_{k+1}^2/ k^4\rp$.
\end{enumerate}
\end{lemma}
Informally, this result  states that, when the underlying input graph $G$ presents a clear structure of clusters, one can find in polynomial-time a partition $\{P_i\}_{i=1}^{\ell}$ such that both the outer and inner conductance of every $P_i$ can be bounded. 
One natural question raising from this partition $\{P_i\}_{i=1}^{\ell}$ is whether we can directly use $\{P_i\}_{i=1}^{\ell}$ to construct an \textsf{HC} tree.  As an obvious approach, one could  
 consider to 
(i) construct trees $\T_i=\T_{\deg}(G[P_i])$ for every $1\leq i\leq \ell$, and 
 (ii) merge the trees $\{\mathcal{\T}_i\}$ in the best way to construct the final tree $\T_G$.
 Unfortunately, as we will see in the following  example, this approach fails to achieve an $O(1)$-approximation.

\begin{example}\label{ex:Well-Clustered Bad Example} We study the following graph $G$, in which all the edges have   unit weight:
\begin{enumerate}
    \item Let $P_1=(V_1, E_1)$ and $P_2=(V_2, E_2)$ be the two copies of the graph constructed in Example~\ref{ex:Expander Bad Example}. 
    Specifically, every $P_i$ is a constant-degree expander graph of   $n$ vertices with an additional planted clique. 
    We use $S_i$ to represent the clique embedded in $P_i$, and $|S_i|=\lfloor n^{2/3} \rfloor$ for 
    $i\in\{1,2\}$;
    \item We define $G\triangleq (V_1\cup V_2, E_1\cup E_2\cup E_3)$, where $E_3$ consists of $\Theta(n^{1.1})$ crossing edges between $S_1$ and $S_2$, see Figure~\ref{Fig:k_clusters_bad_example}(a) for illustration.
\end{enumerate}
By construction, we know that both of $P_1$ and $P_2$ have low outer conductance  $\Phi(P_i) = O(n^{-0.23})$ and high inner conductance $\Phi_{G[P_i]}=\Omega(1)$, 
for $i\in\{1,2\}$. Moreover, it is easy to see that the tree $\T_G$ constructed by merging the trees $\T_1$ and $\T_2$, as shown in Figure~\ref{Fig:k_clusters_bad_example}(b) satisfies that $\cost(\T_G)=\Theta(n^{2.1})$. On the other hand, as illustrated in Figure~\ref{Fig:k_clusters_bad_example}(c), we can place  the two cliques together and further down the  $\textsf{HC}$ tree, and obtain  a tree $\T^{\star}$ with cost $\Theta(n^2)$. Thus we conclude that $\cost(\T_G) = \Omega \lp n^{0.1}\rp \cdot \OPT_G$.
\end{example}
\begin{figure}
    \centering
    \begin{subfigure}[t]{.3\textwidth}
        \centering
        \begin{tikzpicture}[scale=0.4]
            \draw[color=blue, fill=blue!20] (-4,0) circle [x radius=2cm, y radius=3.5cm] node[anchor=south, below=15pt]{$P_1$};
            
            \draw[color=blue, fill=blue!20] (2,0) circle [x radius=2cm, y radius=3.5cm] node[anchor=south, below=15pt]{$P_2$};
            
            \draw (-3.5, 1.5) -- (1.5, 1.5);
            \draw (-3.5, 2) -- (1.5, 2);
            \draw (-3.5, 1) -- (1.5, 1);
            \draw[color=red, fill=red!20] (-4, 1.5) circle [radius=1cm] node {$S_1$};
            \draw[color=red, fill=red!20] (2, 1.5) circle [radius=1cm] node {$S_2$};
        \end{tikzpicture}
        \caption{}
    \end{subfigure}
    \begin{subfigure}[t]{.3\textwidth}
        \centering{}
        \begin{tikzpicture}[scale=0.7, level 1/.style={sibling distance=2cm}, level 2/.style={sibling distance=2cm}] 
            \node [draw, circle,inner sep=1pt, minimum size=1cm]{\scriptsize{$\mathcal{T}_{G}$}}
            child { node [draw, circle, color = orange,fill=orange!20,inner sep=1pt, minimum size=1cm, sibling distance=2cm] {\scriptsize{$\mathcal{T}_1$}} 
            }
            child { node [draw, circle, color = orange,fill=orange!20,inner sep=1pt, minimum size=1cm] {\scriptsize{$\mathcal{T}_2$}} 
            };
        \end{tikzpicture}
        \caption{}
    \end{subfigure}
    \begin{subfigure}[t]{.3\textwidth}
        \centering
        \begin{tikzpicture}[scale=0.7, sibling distance=2cm] 
            \node [draw, circle,inner sep=1pt, minimum size=1cm]{\scriptsize{$\mathcal{T}^*$}}
            child { node [draw, circle, color = blue,fill=blue!20,inner sep=1pt, minimum size=1cm, sibling distance=2cm] {\scriptsize{$P_1\setminus S_1$}}
            }
            child { node [draw, circle,inner sep=1pt, minimum size=1cm] {}
                child { node [draw, circle, color = blue,fill=blue!20,inner sep=1pt, minimum size=1cm] {\scriptsize{$P_2 \setminus S_2$}} 
                }
                child { node [draw, circle,inner sep=1pt, minimum size=1cm] {}
                    child{node [draw, circle, color = red,fill=red!20,inner sep=1pt, minimum size=1cm]{\scriptsize{$S_1$}}
                    }
                    child{node [draw, circle, color = red,fill=red!20,inner sep=1pt, minimum size=1cm]{\scriptsize{$S_2$}}
                    }
                }
            };
        \end{tikzpicture}
        \caption{}
    \end{subfigure}
    \caption{\small{(a) Our constructed graph $G$ consisting of two disjoint copies of the graph in Example~\ref{ex:Expander Bad Example} of $n$ vertices each, with additional $\Theta(n^{1.1})$ crossing edges connecting the two cliques. (b) The tree $\T_G$ formed by merging the two subtrees $\T_1$ and $\T_2$. (c) The tree $\T^*$ constructed in a ``Caterpillar'' fashion by first merging the two cliques $S_1$ and $S_2$.}}
    \label{Fig:k_clusters_bad_example}
\end{figure}
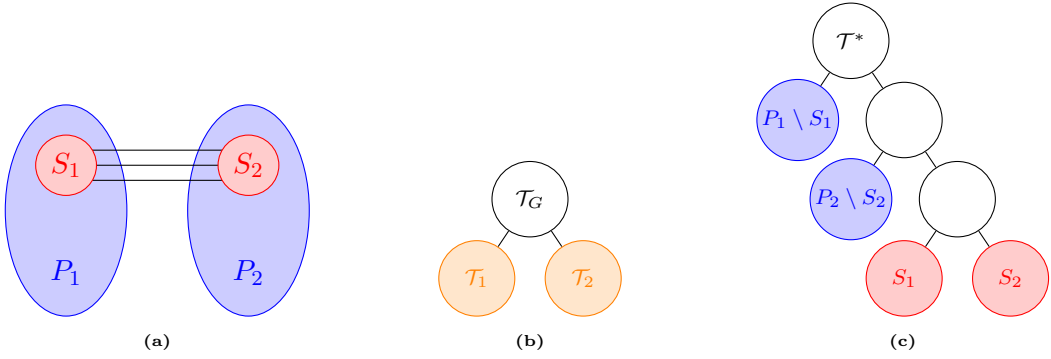

To address this, we follow our intuition gained from Example~\ref{ex:Expander Bad Example}, and further decompose every $P_i$ into smaller subsets. Similar with  analysing dense branches, we introduce the \emph{critical nodes} associated with each $\T_i$.

\begin{definition}[Critical nodes]\label{def:Critical Nodes}
        Let $\T_i = \T_{\deg}(G[P_i])$ be the tree computed by Algorithm~\ref{algo:degree} to the induced graph $G[P_i]$. Suppose $(A_0, \dots, A_{r_i})$ is the dense branch of $\T_i$ for some $r_i \in \mathbb{Z}_{+}$, $B_j$ is the sibling of $A_j$, and let $A_{r_i+1}, B_{r_i+1}$ be the two children of $A_{r_i}$. We define $\mathcal{S}_i \triangleq \{B_1, \dots, B_{r_i + 1}, A_{r_i + 1}\}$ to be the set of \emph{critical nodes} of $P_i$. Each node $N \in \mathcal{S}_i$ is a \emph{critical node}.
\end{definition}
We remark that each critical node $N \in \mathcal{S}_i~(1\leq i\leq \ell)$ is an internal node of maximum size in $\mathcal{T}_i$ that is not in the dense branch.  Moreover, each $\mathcal{S}_i$ is a partition of $P_i$. 
Based on critical nodes,  we  present an improved decomposition algorithm, which is similar to  the one in  Lemma~\ref{lem:Gharan-Trevisan decomposition}, and prove that the output quality of our 
algorithm  can be significantly strengthened for hierarchical clustering. Specifically, in addition to satisfying ($A1$) and ($A2$), we prove that  the total weight between each critical node $N \in \mathcal{S}_i$ and all the other clusters $P_j$ for all $i \neq j$
 can be upper bounded. We highlight that this is one of the key properties that allows us to obtain our main result, and also suggests that the original decomposition algorithm in \cite{GT14} might not suffice for our purpose. 


\begin{lemma}[Strong Decomposition Lemma]\label{lem:Improved Decomposition}
Let $G = (V, E, w)$ be a graph such that $\lambda_{k+1} > 0$  and $\lambda_k < \lp \frac{1}{270 \cdot c_0 \cdot (k+1)^6} \rp^2$, where $c_0$ is the constant in Lemma~\ref{lem:Upperbound eig induced graphs}. Then, there is a polynomial-time algorithm that finds an $\ell$-partition of $V$ into sets $\{P_i\}_{i=1}^{\ell}$, for some $\ell \leq k$, such that for every $1 \leq i \leq \ell$ and every critical node $N \in \mathcal{S}_i$ the following properties hold:
\begin{enumerate}[label=$(A\arabic*)$]
    \item $\Phi(P_i) = O(k^6 \sqrt{\lambda_k})$;
    \item $\Phi_{G[P_i]} = \Omega (\lambda_{k+1}^2/k^4)$;
    \item $w(N, V \setminus P_i) \leq 6(k+1) \cdot \vol_{G[P_i]}(N)$.
\end{enumerate}

\end{lemma}

To underline the importance of ($A3$), recall that, in general, each subtree $\T_i$ cannot be directly used to construct an $O(1)$-approximate $\textsf{HC}$ tree of $G$ because of the potential high cost of the crossing edges $E(P_i, V\setminus P_i)$.
However, if the internal cost of $\T_i$ is high enough to compensate for the cost introduced for the crossing edges $E(P_i, V\setminus P_i)$, then one can safely use this $\T_i$ as a building block. This is one of  the most crucial insights that leads us to design our final algorithm \texttt{PruneMerge}.

\subsection{The algorithm for well-clustered graphs}\label{sec:The algorithm for well-clustered graphs}
Now we are ready to describe  the algorithm \texttt{PruneMerge}, and we refer the reader to Algorithm~\ref{alg:PruneMerge} for the formal presentation.  At a high-level, our algorithm consists of three   phases: \texttt{Partition}, \texttt{Prune} and \texttt{Merge}. In the \texttt{Partition} phase~(Lines \ref{alg:PruneMerge:line:begin partition}--\ref{alg:PruneMerge:line:end partition}), the algorithm invokes Lemma~\ref{lem:Improved Decomposition} to partition  $V(G)$ into sets $\{P_i\}_{i=1}^{\ell}$, and applies Algorithm~\ref{algo:degree} 
to obtain the corresponding trees $\{\T_i\}_{i=1}^{\ell}$. The \texttt{Prune} phase~(Lines \ref{alg:PruneMerge:line:begin prune}--\ref{alg:PruneMerge:line:end prune}) consists of a repeated pruning process: 
for every such tree $\T_i$, the algorithm checks in Line~\ref{alg:PruneMerge:line:if-condition} if the maximal possible cost of the edges $E(P_i, V\setminus P_i)$ (i.e., the LHS of the inequality in the \texttt{if}-condition) can be 
bounded by the internal cost of the critical nodes $N \in S_i$, up to 
 a factor of $O(k).$
\vspace{-0.2cm}
\begin{itemize}\itemsep -1pt
    \item  If so, the algorithm uses $\T_i$ as a building block and  adds it to a global set of trees $\mathbb{T}$;
    \item Otherwise, the algorithm prunes the subtree $\T_i[N]$, where $N \in \mathcal{S}_i$ is the critical node closest to the root in $\T_i$, and adds $\T_i[N]$ to $\mathbb{T}$~(Line \ref{alg:PruneMerge:line:end prune}).
\end{itemize}
 The process is repeated
with the pruned $\T_i$ until either the condition in  Line~\ref{alg:PruneMerge:line:if-condition} is satisfied, or $\T_i$ is completely pruned. Finally, in the \texttt{Merge} phase~(Lines \ref{alg:PruneMerge:line:begin merge}--\ref{alg:PruneMerge:line:end merge})  the algorithm combines the trees in $\mathbb{T}$ in a ``caterpillar style'' according to an increasing order of their sizes. The performance of this algorithm is summarised as follows:

\begin{algorithm}
    \DontPrintSemicolon
    
    \KwInput{A graph $G = (V, E, w)$, a parameter $k \in \mathbb{Z}_{+}$ such that $\lambda_{k+1} > 0$;}
    \KwOutput{An \textsf{HC} tree $\TPrune$ of $G$;}
    
    Apply the partitioning algorithm (Lemma~\ref{lem:Improved Decomposition}) on input $(G, k)$ to obtain $\{P_i\}_{i=1}^{\ell}$ for some $\ell \leq k$;\label{alg:PruneMerge:line:begin partition}
    
    Let $\T_i = \mathrm{\texttt{HCwithDegrees}}(G[P_i])$;\label{alg:PruneMerge:line:end partition}

    Initialise $\mathbb{T} = \emptyset$;
    
    \For{All clusters $P_i$}
    {
        \label{alg:PruneMerge:line:begin prune}
        Let $\mathcal{S}_i$ be the set of critical nodes of $\T_i$;
        
        \While{ $\mathcal{S}_i$ is nonempty}
        {
            \If{ $n \cdot \sum_{N \in \mathcal{S}_i} w(N, V \setminus P_i) \leq 6(k+1) \cdot \sum_{N \in \mathcal{S}_i} |\parent_{\T_i}(N)| \cdot \vol_{G[P_i]}(N)$}
            {
                \label{alg:PruneMerge:line:if-condition}
                Update $\mathbb{T} \leftarrow \mathbb{T} \cup \T_i$ and $\mathcal{S}_i = \emptyset$;
            }
            \Else
            {
                Let $N, M$ be the two children of the root of $\mathcal{T}_i$ such that $N \in \mathcal{S}_i$;\label{alg:PruneMerge:line:prune step}
                
                Update $\mathbb{T} \leftarrow \mathbb{T} \cup \T_i[N]$, 
                $\mathcal{S}_i \leftarrow \mathcal{S}_i \setminus \{N\}$ 
                and $\T_i \leftarrow \T_i[M]$;\label{alg:PruneMerge:line:end prune}
            }
            
        }
    }
    
    Let $t = |\mathbb{T}|$ and $\mathbb{T} = \{\widetilde{\T_1}, \dots, \widetilde{\T_t}\}$ be such that $|\widetilde{\T}_i| \leq |\widetilde{\T}_{i+1}|$, for all $1 \leq i < t$;
    
    Initialise $\TPrune = \widetilde{\T}_1$;\label{alg:PruneMerge:line:begin merge}
    
    \For{ $i=2,\ldots, t$}
    {
        Let $\TPrune$ be the tree with $\TPrune$ and $\widetilde{\T}_i$ as its two children;\label{alg:PruneMerge:line:end merge}
    }
    
    \Return $\TPrune$

\caption{\texttt{PruneMerge}$(G, k)$\label{alg:PruneMerge}}
\end{algorithm}

\begin{theorem}\label{thm:main_k_clusters}
Let $G=(V, E, w)$ be a  graph, and $k>1$  such that $\lambda_{k+1}>0$  and $\lambda_k < \lp \frac{1}{270 \cdot c_0 \cdot (k+1)^6} \rp^2$, where $c_0$ is the constant in Lemma~\ref{lem:Upperbound eig induced graphs}. The algorithm 
\emph{\texttt{PruneMerge}}  runs in  polynomial-time and   constructs an \textsf{HC} tree $\TPrune$ of $G$ satisfying $\cost_G(\TPrune) = O\lp k^{22}/\lambda_{k+1}^{10}\rp \cdot \OPT_G$. In particular, when  $\lambda_{k+1}=\Omega(1)$ and $k=O(1)$, the algorithm's constructed tree $\TPrune$ satisfies that $\cost_G(\TPrune) = O(1) \cdot \OPT_G$.  
\end{theorem}

We remark that, although Algorithm~\ref{alg:PruneMerge} requires a parameter $k$ as input, we can apply the standard   technique of running Algorithm~\ref{alg:PruneMerge} for different values of $k$ and return the tree of lowest cost. 
By introducing a factor of $O(k)$ to the algorithm's runtime, this ensures that one of the constructed trees by Algorithm~\ref{alg:PruneMerge} would always satisfy our promised approximation ratio.

\subsection{Analysis of the algorithm}\label{sec:The analysis of the algorithm}

In this subsection  we will prove our  main result, i.e.,  Theorem~\ref{thm:main_k_clusters}. We assume that the \textsf{Partition} phase of Algorithm~\ref{alg:PruneMerge} (Line~\ref{alg:PruneMerge:line:begin partition}) has finished and   $V(G)$ is decomposed into disjoint sets $\{P_i\}_{i=1}^{\ell}$, for some $\ell \leq k$, such that the following properties of Lemma~\ref{lem:Improved Decomposition} hold for all $1 \leq i \leq \ell$ and every critical node $N \in \mathcal{S}_i$: 
\begin{enumerate}[label=$(A\arabic*)$]
    \item $\Phi(P_i) = O(k^6 \sqrt{\lambda_k})$; 
    \item $\Phi_{G[P_i]} = \Omega (\lambda_{k+1}^2/k^4)$;
    \item $w(N, V \setminus P_i) \leq 6(k+1) \cdot \vol_{G[P_i]}(N)$.
\end{enumerate}
Also,  for some parameter  $ \phiIn = \Theta\lp\lambda_{k+1}/k^2\rp$, we have that $\Phi_{G[P_i]} = \Omega(\phiIn^2)$ holds for any $1\leq i\leq \ell$. Let $\T_i = \T_{\deg}(G[P_i])$ with corresponding  set  of critical nodes $\mathcal{S}_i$, for all $1 \leq i \leq \ell$, and let $\mathcal{S} = \bigcup_{i=1}^{\ell} \mathcal{S}_i$ be the set of all critical nodes. In order to prove Theorem~\ref{thm:main_k_clusters},  we adopt the following strategy. We group the edges of $G$ into two categories: let $E_1$ be the set of edges in the induced subgraphs $G[P_i]$ for all $1\leq i\leq \ell$, i.e.,
\[
    E_ 1 \triangleq  \bigcup_{i=1}^{\ell} E\left[G[P_i] \right],
\]
and let $E_2$ be the remaining crossing edges. Therefore, we can write the cost of our tree $\TPrune$ as
\begin{equation}\label{eq:Split cost into E_1 and E_2}
    \cost(\TPrune) = \sum_{e \in E_1} \cost_{\TPrune}(e) + \sum_{e \in E_2} \cost_{\TPrune}(e),
\end{equation}
and we will ultimately bound each sum individually in Lemmas~\ref{lem:cost edges inside each cluster} and \ref{lem:cost crossing edges}. However, before presenting these two technical lemmas, we first  introduce   the notions of \emph{pruned} and \emph{unpruned} critical nodes, and analyse their properties.

Consider an arbitrary cluster $P_i$ with corresponding induced tree $\T_i$. We say that a critical node $N \in \mathcal{S}_i$ is \emph{pruned} if $N$ was eventually cut from the tree $\T_i$ and the subtree $\T_i[N]$ was added to $\mathbb{T}$ (Lines~\ref{alg:PruneMerge:line:prune step}--\ref{alg:PruneMerge:line:end prune} of Algorithm~\ref{alg:PruneMerge}). Otherwise, we say that $N$ is \emph{unpruned}. We denote the set of all pruned nodes by \textsc{Pruned}
and the set of unpruned nodes by \textsc{Unpruned}.

Our first result bounds the size of the parent of a pruned node $N \in \mathcal{S}_i$ in $\TPrune$, with respect to the size of its parent in the tree $\T_i$. This result will be extensively used when bounding the cost of the edges adjacent to $N$.

\begin{lemma}\label{lem:Parent size pruned nodes}
It holds for every $1 \leq i \leq \ell$ and every pruned critical node $N \in \mathcal{S}_i \cap \mathrm{\textsc{Pruned}}$ that 
    \[
        \left|\parent_{\TPrune}(N)\right| \leq 6k \cdot \left|\parent_{\T_i}(N)\right|.
    \]  
\end{lemma}
 \begin{proof}
    Suppose the dense branch of $\T_i$ is $(A_0, \dots, A_{k_i})$ for some $k_i \in \mathbb{Z}_{\geq 0}$, with $B_j$ being the sibling of $A_j$ and $A_{k_i}$ having children $A_{k_i +1}, B_{k_i +1}$. Recall that the set of critical nodes is $\mathcal{S}_i = \{ B_1, \dots, B_{k_i +1}, A_{k_i +1}\}$. By construction, it holds for all $1 \leq j \leq k_i$ that  $|A_j| = 2 \cdot |A_{j+1}|$, which implies that  $|B_{j + 1}| = |A_{j+1}|$ and  $|B_{j}| = 2 \cdot |B_{j+1}|$ for all $j \geq 2$. Thus, we   conclude that for every interval $\lp 2^{s-1}, 2^s\rsp$, for some $s \in \mathbb{Z}_{\geq 0}$, there are at most $3$ critical nodes\footnote{We remark that, in the worst case, all three nodes $B_1, B_{k_i + 1}$ and $A_{k_i + 1}$ could have size in $(2^{s-1}, 2^s]$. } $N \in \mathcal{S}_i$ of size $|N| \in \lp 2^{s-1}, 2^s \rsp$. 
    Now let us fix $N \in \mathcal{S}_i \cap \mathrm{\textsc{Pruned}}$. By construction, we have that
    \begin{equation} \label{eq:claim parent size eq1-new}
        \left|\parent_{\T_i}(N)\right| \geq 2 \cdot |N|.
    \end{equation}
    On the other hand, by the construction of $\TPrune$ we have that 
    \begin{align}
        \left|\parent_{\TPrune}(N)\right| 
        &= \sum_{j=1}^{\ell} \sum_{\substack{M \in \mathcal{S}_j \\ |M| \leq |N|}} |M|
        \leq \sum_{j=1}^{\ell} \sum_{s=0}^{\ceil{\log{|N|}}} \sum_{\substack{M \in \mathcal{S}_j \\ 2^{s-1} < |M| \leq 2^s}} |M| \nonumber\\
        &\leq \sum_{j=1}^{\ell} \sum_{s = 0}^{\ceil{\log{|N|}}} 3 \cdot 2^s
        \leq \sum_{j=1}^k 3 \cdot 2^{\ceil{\log{|N|}} + 1} \label{eq:claim parent size eq2-new}
        \leq 12k \cdot |N|.
    \end{align}
    By combining \eqref{eq:claim parent size eq1-new} and \eqref{eq:claim parent size eq2-new}, we have  that 
    \[
        \left|\parent_{\TPrune}(N)\right| \leq 12k \cdot |N| \leq 6k \cdot \left|\parent_{\T_i}(N)\right|,
    \]
    which proves the statement.
\end{proof}

We are now ready to prove the two main technical lemmas of this subsection.

\begin{lemma}\label{lem:cost edges inside each cluster}
    It holds that $\sum_{e \in E_1} \cost_{\TPrune}(e) = O\left( k/\phiIn^8 \right) \cdot \OPT_G$.
\end{lemma}

\begin{proof}
    Notice that   
    \begin{equation*}\label{eq:lem:cost edges inside each cluster eq1}
        \sum_{e \in E_1} \cost_{\TPrune}(e) = \sum_{i=1}^{\ell} \sum_{e \in E[G[P_i]]} \cost_{\TPrune}(e). 
    \end{equation*}
    
    We will prove that, for every $1\leq i\leq\ell$ and $e\in E[G[P_i]]$, the cost of $e$ in $\T_G$ and the one in $\T_i$ differ by at most a factor of $O(k)$. Combining this with Theorem~\ref{thm:degree}  will prove the lemma.
    
    To prove this $O(k)$-factor bound, we fix any $1\leq i\leq \ell$ and let $\mathcal{S}_i$ be the set of critical nodes of $P_i$. As the nodes of $\mathcal{S}_i$ form a partition of the vertices of $G[P_i]$, any edge $e\in E[G[P_i]]$  satisfies exactly one of the following conditions: (i) $e$ is inside a critical node; (ii) $e$ is adjacent to a pruned node; or (iii) $e$ crosses two unpruned nodes. Formally, it holds that 
    
    \begin{align*}
         \sum_{e \in E[G[P_i]]} \cost_{\TPrune}(e) =  
        &\sum_{\substack{N \in \mathcal{S}_i \\ e \in E(N,N)}} \cost_{\TPrune}(e) 
        + \sum_{\substack{N \in \mathcal{S}_i \cap \mathrm{\textsc{Pruned}} \\ M \in \mathcal{S}_i \setminus \{N\} \\ \left|\parent_{\T_i}(M) \right| \leq \left|\parent_{\T_i}(N) \right|}} \sum_{e \in E(N, M)} \cost_{\TPrune}(e) 
        \\
       &\qquad \qquad+ \sum_{\substack{N, M \in \mathcal{S}_i \cap \mathrm{\textsc{Unpruned}} \\ N \neq M}} \sum_{e \in E(N, M)} \cost_{\TPrune}(e)
    \end{align*}
    
    For Cases~(i)~and~(iii), the costs of  $e$ in both trees are the same, since we do not change the structure of the tree inside any critical node nor alter the inner structure of the pruned trees $\T_i$ that contain only unpruned nodes, i.e.,
    \begin{align*}
        \sum_{\substack{N \in \mathcal{S}_i \\ e \in E(N,N)}} \cost_{\TPrune}(e) &= \sum_{\substack{N \in \mathcal{S}_i \\ e \in E(N,N)}} \cost_{\T_i}(e), \\
        \sum_{\substack{N, M \in \mathcal{S}_i \cap \mathrm{\textsc{Unpruned}} \\ N \neq M}} \sum_{e \in E(N, M)} \cost_{\TPrune}(e) &=  \sum_{\substack{N, M \in \mathcal{S}_i \cap \mathrm{\textsc{Unpruned}} \\ N \neq M}} \sum_{e \in E(N, M)} \cost_{\T_i}(e).
    \end{align*}
  For Case~(ii), the cost of any such edge increases by at most a factor of  $O(k)$  due to Lemma~\ref{lem:Parent size pruned nodes} and the construction of $\TPrune$. 
    Formally, let $N \in \mathcal{S}_i \cap \textsc{Pruned}$ be an arbitrary pruned node and let $M \in \mathcal{S}_i \setminus\{N\}$ be a critical node such that $\left|\parent_{\T_i}(M) \right| \leq \left|\parent_{\T_i}(N) \right|$. Firstly, notice that if $\parent_{\T_i}(N)$ is the root node of $\T_i$, then for any edge $e \in E(N, M)$ it holds that $\cost_{\T_i}(e) = w_e \cdot |\T_i|.$ On the other side, by the construction of $\TPrune$, we know that $\cost_{\TPrune}(e) \leq 6k \cdot w_e \cdot |\T_i|$, so we conclude that $\cost_{\TPrune}(e) \leq 6k \cdot \cost_{\T_i}(e)$. Secondly, if $\parent_{\T_i}(N)$ is not the root node of $\T_i$ and since $\left|\parent_{\T_i}(M) \right| \leq \left|\parent_{\T_i}(N) \right|$, we know that $|M| \leq |N|$. Therefore it holds for any edge $e \in E(N, M)$ that
    \[
        \cost_{\TPrune}(e) 
        = w_e \cdot \left | \parent_{\TPrune}(N)\right |
        \leq 6k \cdot w_e \cdot \left | \parent_{\T_i}(N)\right |
        = 6k \cdot \cost_{\T_i}(e),
    \]
    where the inequality follows by Lemma~\ref{lem:Parent size pruned nodes}.
    Combining the above observations, we have  that 
    \begin{equation}\label{eq:lem:cost edges inside each cluster eq3-new}
        \sum_{e \in E_1} \cost_{\TPrune}(e)
        \leq \sum_{i=1}^{\ell} 6k \cdot \sum_{e \in E[G[P_i]]} \cost_{\T_i}(e)
        \leq \sum_{i=1}^{\ell} 6k \cdot \cost(\T_i).
    \end{equation}
    On the other side, let $\T^*$ be any optimal \textsf{HC} tree of $G$ with cost $\OPT_G$, and it holds that 
    \begin{align}\label{eq:lem:cost edges inside each cluster eq4-new}
        \OPT_G &= \cost_G(\T^*) 
        \geq \sum_{i=1}^{\ell} \sum_{e \in E(G[P_i])} \cost_{\T^*}(e)\geq \sum_{i=1}^{\ell} \OPT_{G[P_i]}
       \nonumber \\
        &= \sum_{i=1}^{\ell} \Omega \lp \Phi_{G[P_i]}^4\rp \cdot \cost(\T_i) 
        = \sum_{i=1}^{\ell} \Omega \lp \phiIn^8 \rp \cdot \cost(\T_i),
    \end{align}
    where the last equality follows by Property~$(A2)$  of Lemma~\ref{lem:Improved Decomposition} and Theorem~\ref{thm:degree} applied to every $G[P_i]$. Finally, by combining  \eqref{eq:lem:cost edges inside each cluster eq3-new} and \eqref{eq:lem:cost edges inside each cluster eq4-new} we have that 
    \[
        \sum_{e \in E_1} \cost_{\TPrune}(e) = O(k/\phiIn^8) \cdot \OPT_G,
    \]
    which proves the lemma.
\end{proof}

\begin{lemma}\label{lem:cost crossing edges}
It holds that     $\sum_{e \in E_2} \cost_{\TPrune}(e) = O\left( k^2/\phiIn^{10} \right) \cdot \OPT_G$.
\end{lemma}
\begin{proof}
    For the edges $e \in E_2$, we can subdivide them into (i) edges adjacent to pruned nodes and (ii) edges adjacent to unpruned nodes. For Case~(i), we will bound the cost with the help of Lemma~\ref{lem:Parent size pruned nodes} similar as before. For Case~(ii) we will upper bound the cost based on the fact that, for unpruned nodes, the condition in Line~\ref{alg:PruneMerge:line:if-condition} of Algorithm~\ref{alg:PruneMerge} is satisfied. Specifically, we have that 
    \begin{align}
        &\sum_{e \in E_2} \cost_{\TPrune}(e) \nonumber 
        \leq \sum_{i=1}^{\ell} \sum_{\substack{N \in \mathcal{S}_i \\ M \in \mathcal{S} \setminus \mathcal{S}_i \\ |M| \leq |N|}} \sum_{\substack{e \in E(N, M)}} \cost_{\TPrune}(e) \nonumber\\
        &\leq \sum_{i=1}^{\ell} \sum_{\substack{N \in \mathcal{S}_i \cap \mathrm{\textsc{Unpruned}}\\ M \in \mathcal{S} \setminus \mathcal{S}_i \\ |M| \leq |N|}} \sum_{e \in E(N, M)} \cost_{\TPrune}(e)
        + \sum_{i=1}^{\ell} \sum_{\substack{N \in \mathcal{S}_i \cap \mathrm{\textsc{Pruned}}\\ M \in \mathcal{S} \setminus \mathcal{S}_i \\ |M| \leq |N|}} \sum_{e \in E(N, M)} \cost_{\TPrune}(e) \nonumber \\
        &\leq \sum_{i=1}^{\ell} \sum_{N \in \mathcal{S}_i \cap \mathrm{\textsc{Unpruned}}} n \cdot w(N, V \setminus P_i)
        + \sum_{i=1}^{\ell} \sum_{N \in \mathcal{S}_i \cap \mathrm{\textsc{Pruned}}}  \left| \parent_{\TPrune}(N)\right| \cdot w(N, V\setminus P_i) \nonumber\\
        &\leq \sum_{i=1}^{\ell} \sum_{N \in \mathcal{S}_i \cap \mathrm{\textsc{Unpruned}}} 6 (k+1) \cdot |\parent_{\T_i}(N)|\cdot \vol_{G[P_i]}(N) \label{eq:lem:cost crossing edges eq1}\\
        & \qquad + \sum_{i=1}^{\ell} \sum_{N \in \mathcal{S}_i \cap \mathrm{\textsc{Pruned}}}  36 k(k+1) \cdot\left| \parent_{\T_i}(N)\right| \cdot \vol_{G[P_i]}(N) \label{eq:lem:cost crossing edges eq2}\\
        & \leq 36 k(k+1) \cdot \sum_{i=1}^{\ell} \sum_{N \in \mathcal{S}_i} \left| \parent_{\T_i}(N)\right| \cdot \vol_{G[P_i]}(N) \nonumber\\
        & \leq 36k(k+1) \cdot \sum_{i=1}^{\ell} \frac{4}{\Phi_{G[P_i]}} \cdot \cost_{G[P_i]}(\T_i) \label{eq:lem:cost crossing edges eq3}\\
        &= O(k^2) \cdot \sum_{i=1}^{\ell} \frac{1}{\Phi_{G[P_i]}^5} \cdot \OPT_{G[P_i]} \label{eq:lem:cost crossing edges eq4}\\
        & = O(k^2/\phiIn^{10}) \cdot \OPT_G, \nonumber
    \end{align}
    where \eqref{eq:lem:cost crossing edges eq1} follows by fact that the   unpruned nodes satisfy   the \texttt{if}-condition in Line~\ref{alg:PruneMerge:line:if-condition} of Algorithm~\ref{alg:PruneMerge}, \eqref{eq:lem:cost crossing edges eq2} follows from Property $(A3)$ of Lemma~\ref{lem:Improved Decomposition} and Lemma~\ref{lem:Parent size pruned nodes}, \eqref{eq:lem:cost crossing edges eq3} follows by Lemma~\ref{lem:Cost light nodes}, and \eqref{eq:lem:cost crossing edges eq4} follows by Theorem~\ref{thm:degree} applied to every induced subgraph $G[P_i]$. 
\end{proof}

Finally, we are ready to prove  Theorem~\ref{thm:main_k_clusters}.

\begin{proof}[Proof of Theorem~\ref{thm:main_k_clusters}]
    Let $\TPrune$ be the \textsf{HC} tree obtained from Algorithm~\ref{alg:PruneMerge}. We have that
    \begin{align*}
        \lefteqn{\cost(\TPrune)}\\
        &= \sum_{e \in E_1} \cost_{\TPrune}(e) + \sum_{e \in E_2} \cost_{\TPrune}(e)\\
        &= O(k/\phiIn^8) \cdot \OPT_G + O(k^2/\phiIn^{10}) \cdot \OPT_G
        = O(k^2/\phiIn^{10}) \cdot \OPT_G= O(k^{22}/\lambda_{k+1}^{10}) \cdot \OPT_G,
    \end{align*}
    where   the second equality follows by Lemmas~\ref{lem:cost edges inside each cluster} and \ref{lem:cost crossing edges}, and the last equality follows by the definition of $\phiIn$.
    
    Next we  analyse the runtime of our algorithm. The \textsf{Partition} phase~(Lines~\ref{alg:PruneMerge:line:begin partition}--\ref{alg:PruneMerge:line:end partition}) runs in polynomial time by Lemma~\ref{lem:Improved Decomposition}. In the \textsf{Prune} phase~(Lines~\ref{alg:PruneMerge:line:begin prune}--\ref{alg:PruneMerge:line:end prune}), the algorithm goes through all $\ell \leq k$ trees $\T_i$, and for each tree the algorithm attempts to prune the critical node closest to the root. If the pruning happens,  the process is repeated recursively to the pruned tree. Since there are $O(\log n)$ critical nodes for each tree $\T_i$ by construction,  the \texttt{if}-condition in Line~\ref{alg:PruneMerge:line:if-condition} of Algorithm~\ref{alg:PruneMerge} will be checked    $O(\log n)$  number of times  for each tree $\T_i$.  Hence  we conclude that the \textsf{Prune} phase of the algorithm runs in polynomial time. Finally, in the \textsf{Merge} phase~(Lines~\ref{alg:PruneMerge:line:begin merge}--\ref{alg:PruneMerge:line:end merge}) the algorithm goes through the global set of trees $\mathbb{T}$, and successively merges the trees $\widetilde{\T_i} \in \mathbb{T}$ to form the final tree $\TPrune$. Since $|\mathbb{T}| = O(\ell \cdot \log n)$, we conclude that this step will be finished in polynomial time.
\end{proof}

\subsection{Proof of Lemma~\ref{lem:Improved Decomposition}}\label{sec:Proof of Improved Decomp Lemma}

Finally, we finish the section by presenting the proof of Lemma~\ref{lem:Improved Decomposition}.   We first describe the underlying algorithm and show a sequence of claims, which are used to prove Lemma~\ref{lem:Improved Decomposition} in the end of the subsection. At the very high level, our  algorithm for computing a stronger decomposition of a   well-clustered graph can be viewed as an extension to Algorithm~3 in \cite{GT14}, whose main idea 
 can be summarised as follows: the algorithm starts with the trivial $1$-partition of $G$, i.e.,  $P_1 = V$; in every iteration, the algorithm applies  the \SpecPart algorithm for every graph in   $\{G[P_i]\}_{i=1}^{r}$, and tries to find a sparse cut $(S, P_i \setminus S)$ for some  $S \subset P_i$.\footnote{We will denote by $r$ the number of clusters in the current run of the algorithm, and   denote by $\ell$ the final number of clusters output by the algorithm.}   
\begin{itemize}
    \item If such a cut is found, the algorithm uses this cut to either introduce a new partition set $P_{r+1}$ of small conductance, or   refine the current  partition  $\{P_i\}_{i=1}^r$;
    
    \item  If no such cut is found, the algorithm checks if it is possible to perform a local refinement of the partition sets $\{P_i\}_{i=1}^r$ in order to reduce the overall weight of the crossing edges, i.e. $\sum_{i \neq j} w(P_i, P_j)$. If such a refinement is not possible, the algorithm terminates and outputs the current partition;  otherwise, the partition sets are locally refined and the process is repeated. 
\end{itemize}
The output of the algorithm is guaranteed to  satisfy Properties $(A1)$ and $(A2)$  of Lemma~\ref{lem:Improved Decomposition}.

Our improved  analysis will show that  Property~$(A3)$ holds as well, and this will be proven  with  the two additional Properties  $(A4)$ and $(A5)$ stated later. 
We begin our analysis by setting the notation, most of which follows from \cite{GT14}. 
We write $\{P_i\}_{i=1}^r$ as  a partition of $V$ for some integer $r\geq 1$, i.e., $P_i\cap P_j=\emptyset$ for $i\neq j$, and $\cup_{i=1}^r P_i=V $. In addition,  every  partition set  $P_i$ contains some \emph{core set} denoted by $\mathrm{core}(P_i)\subseteq P_i$.  For an arbitrary subset $S \subset P_i$,  we define $S^+ \triangleq S\cap \mathrm{core}(P_i)$, and $S^- \triangleq S\setminus S^+$. 
We further define  $\overline{S^+}\triangleq \mathrm{core}(P_i)\setminus S$, and $\overline{S^-}\triangleq P_i\setminus (S \cup \mathrm{core}(P_i))$,  as illustrated in Figure~\ref{fig:Partition_Explanatory}. Note that $\{S^{+}, \overline{S^{+}}\}$ forms a partition of $\mathrm{core}(P_i)$, and $\{S^{-}, \overline{S^{-}}\}$ forms a partition of $P_i \setminus \mathrm{core}(P_i)$.  For any sets $S, T \subseteq V$ which are not necessarily disjoint,  we write 
    \[
        w(S \rightarrow T) \triangleq w(S, T \setminus S).
    \]
    For any subsets $S \subseteq P \subseteq V$, we follow \cite{GT14} and define the \emph{relative conductance} as
    \[
        \varphi(S, P) \triangleq \frac{w(S \rightarrow P)}{\frac{\vol(P \setminus S)}{\vol(P)} \cdot w(S \rightarrow V\setminus P)},
    \]
    whenever the right hand side is defined and otherwise we set $\varphi(S, P) = 1$.
  To explain the meaning of $\varphi(S, P)$, suppose that $P\subset V$ is the vertex set such that $\Phi_G(P)$ is low and $\Phi_{G[P]}$ is high, i.e., $P$ is a cluster. Then, we know that most of the  subsets $S\subset P$ with $\vol(S)\leq \vol(P)/2$ satisfy the following properties:
\begin{itemize}
    \item Since $\Phi_{G[P]}(S)$ is high, a large fraction of the edges adjacent to vertices in $S$ would leave $S$;
    \item Since $\Phi_G(P)$ is low, a small fraction of edges adjacent to $S$ would leave $P$.
\end{itemize}
Combining the above observations, one could conclude that $w(S \rightarrow P) \gtrsim w(S \rightarrow V \setminus P)$ if $P$ is a good cluster, which means that $\varphi(S, P)$ is lower bounded by a constant. Moreover, Gharan and Trevisan~\cite{GT14} showed a converse of this fact:  if $\varphi(S, P)$ is large for all $S \subset P$, then $P$ has high inner conductance. These facts suggest that the relative conductance provides a good quantitative measure for the quality of a cluster.

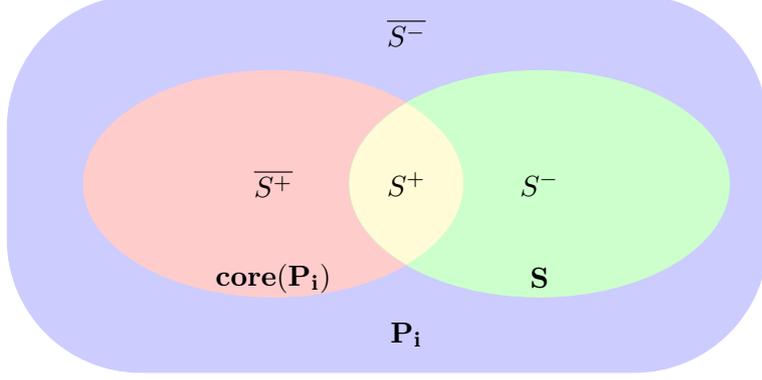
\begin{figure}
    \centering
    \def \P{((0, 0) rectangle (10, 5)}
    \def \Core{(3.5, 2.5) ellipse (2.5cm and 1.5cm)}
    \def \S{(7, 2.5) ellipse (2.5cm and 1.5cm)}

    \begin{tikzpicture}
        \filldraw[blue!20, rounded corners=50pt] \P node[black] {};
        \filldraw[red!20] \Core node[black] {$\overline{S^{+}}$};
        \filldraw[green!20] \S node[black] {$S^-$};
        \begin{scope}
            \clip \Core;
            \fill[yellow!20] \S;
        \end{scope}
        
        \draw node[black] at (5.25, 0.5) {$\mathbf{P_i}$};
        \draw node[black] at (5.25, 4.5) {$\overline{S^-}$};
        \draw node[black] at (5.25, 2.5) {$S^+$};
        \draw node[black] at (7, 1.25) {$\mathbf{S}$};
        \draw node[black] at (3.5, 1.25) {$\mathrm{\textbf{core}} \mathbf{(P_i)}$};
    \end{tikzpicture}
    
    \caption{ Illustration of  different subsets of  $P_i$. The left ellipsoid represents the core set $\mathrm{core}(P_i)$, and the right ellipsoid represents the subset $S$. The subset $S$ is partitioned into $\{S^+, S^-\}$, where $S^+ \triangleq S \cap \mathrm{core}(P_i)$; the set $P_i \setminus S$ is partitioned into $\{\overline{S^+}, \overline{S^-}\}$, where $\overline{S^+} \triangleq \mathrm{core}(P_i) \setminus S$. }
    \label{fig:Partition_Explanatory}
\end{figure}

 Now we explain the high-level idea of the proposed algorithm, and refer the reader to Algorithm~\ref{alg:Strong Decomposition} for the formal description. Our algorithm starts with the partitioning algorithm~(Algorithm~3 in \cite{GT14}), and obtains an   intermediate partition $\{P_i\}_{i=1}^r$~(Lines~\ref{alg:while line}--\ref{alg:line:update if-5}).  For every $P_i~(1\leq i\leq r)$, the algorithm computes the induced trees $\T_i = \T_{\deg}(G[P_i])$~(Line~\ref{alg:line:compute trees}).  For every tree $\T_i~(1\leq i\leq r)$ with  the corresponding set of critical nodes $\mathcal{S}_i$, the algorithm further checks if the following conditions are satisfied:  
\begin{enumerate}[label=$(A\arabic*)$]
    \setcounter{enumi}{3}
    \item  For every critical node $N \in \mathcal{S}_i$ with $\vol(N^+) \leq \vol(\core(P_i))/2$, it holds that $$\varphi(N^+, \core(P_i)) \geq \frac{1}{3(k+1)};$$ 
    \item  For every critical node $N \in \mathcal{S}_i$ with  $\vol(N^-) \leq \vol(P_i)/2$, it holds that $$w(N^- \rightarrow P_i) \geq w(N^- \rightarrow V \setminus P_i) \cdot \frac{1}{k+1}.$$
\end{enumerate}
If $(A4)$ is violated by some critical node $N \in \mathcal{S}_i$ for some $i$, then the algorithm
uses  $N$  to refine the core set $\mathrm{core}(P_i)$~(Line~\ref{alg:line:update if-7}). If $(A5)$ is not satisfied, then the algorithm further refines the partition~(Line~\ref{alg:line:update if-8}).  The algorithm repeats this local refinement process until no such update is found anymore. 
In the following analysis, we set
\[
    \rho^* \triangleq \min \left\{ \frac{\lambda_{k+1}}{10}, 30 c_0 \cdot (k+1)^5 \cdot \sqrt{\lambda_k}\right\},
\]
where $c_0$ is the constant   specified in Lemma~\ref{lem:Upperbound eig induced graphs}, and
\begin{equation}\label{eq:PhiIn PhiOut Choices}
    \phiIn \triangleq \frac{\lambda_{k+1}}{140(k+1)^2}, \qquad
    \phiOut \triangleq 90 c_0 \cdot (k+1)^6 \sqrt{\lambda_k}.
\end{equation}
Notice that, by assuming $\lambda_k < \lp \frac{1}{270 \cdot c_0 \cdot (k+1)^6} \rp^2$ in Lemma~\ref{lem:Improved Decomposition}, it holds that $\phiOut < 1/3$. This fact will be used several times in our analysis.

\begin{algorithm}[H]
    \KwInput{$G = (V, E, w)$, $k > 1$ such that $\lambda_{k+1} > 0$;}
    \KwOutput{A $\lp \phiIn^2/4, \phiOut \rp$ $\ell$-partition $\{P_i\}_{i=1}^{\ell}$ of $G$ satisfying $(A1) - (A3)$, for some $\ell \leq k$;}
  
	Let $r = 1$, $\core(P_1) = P_1 = V$;
	
	Let $\phiIn = \frac{\lambda_{k+1}}{140(k+1)^2}$,  and 
    $\phiOut = 90 c_0 \cdot (k+1)^6 \sqrt{\lambda_k}$;
	
	Let $\rho^* = \min \left\{ \frac{\lambda_{k+1}}{10}, 30 c_0 \cdot (k+1)^5 \cdot \sqrt{\lambda_k}\right\}$;
	
	\While { At least one of the following conditions holds \label{alg:while line}
	    \begin{enumerate}
    	    \item $\exists 1 \leq i \leq r$ such that $w(P_i \setminus \core(P_i) \rightarrow P_i) < w(P_i \setminus \core(P_i) \rightarrow P_j)$ for some $j \neq i$;
    	    \item \SpecPart finds $S \subseteq P_i$ with~\footnotemark{} $\vol(S^+) \leq \vol(\mathrm{core}(P_i))/2$, such that $\max \left\{\Phi_{G[P_i]}(S),  \Phi_{G[P_i]}(P_i \setminus S)\right\} < \phiIn$;
	    \end{enumerate}
    }
    {
        \label{alg:line:beginning while loop}
          Order the sets $P_1,\ldots, P_r$ such that   $\lambda_2(G[P_1]) \leq \ldots \leq \lambda_2(G[P_r])$;
        
        Let $1 \leq i \leq r$ the smallest index for which item~$2$ of the \texttt{while}-condition is satisfied, and let $S \subset P_i$ be the corresponding set;
	    
	    \If{$\max \left\{ \Phi_G(S^+), \Phi_G\left(\barr{S^+}\right)\right\} \leq  \lp 1 + \frac{1}{k+1} \rp^{r+1} \cdot \rho^*$}
	    {
	        \label{alg:line:condition if-1}
	        Let $P_i = P_i \setminus \overline{{S^+}}$, $\core(P_i) = S^+$, $P_{r+1} = \core(P_{r+1}) = \barr{S^+}$,   $r = r + 1$ and \textbf{go to} Line~\ref{alg:while line}\label{alg:line:update if-1};
	    }
	    
	    \If{$ \min \left\{ \varphi(S^+, \core(P_i)), \varphi\left(\barr{S^+}, \core(P_i)\right) \right\} \leq \frac{1}{3(k+1)}$}
	    {
	        \label{alg:line:condition if-2}
	        Update $\core(P_i)$ to either $S^+$ or $\barr{S^+}$ with the lower conductance, and \textbf{go to} Line~\ref{alg:while line}; \label{alg:line:update if-2}
	    }
	    
	    \If{$ \Phi_G(S^-) \leq \lp 1 + \frac{1}{k+1} \rp^{r+1} \cdot \rho^*$}
	    {
	        \label{alg:line:condition if-3}
	        Let $P_i = P_i \setminus S^-$, $P_{r +1} = \core(P_{r+1}) = S^-$, set $r = r +1$ and \textbf{go to} Line~\ref{alg:while line}; \label{alg:line:update if-3}
	    }
	    
	    \If{$ w(P_i \setminus \core(P_i) \rightarrow P_i) < w(P_i \setminus \core(P_i) \rightarrow P_j)$ for some $j \neq i$}
	    {
	        \label{alg:line:condition if-4}
	        Let $P_i = \core(P_i)$, merge $(P_i \setminus \core(P_i))$ with $\mathrm{argmax}_{P_j} \{ w(P_i \setminus \core(P_i) \rightarrow P_j)\}$, and \textbf{go to} Line~\ref{alg:while line}; \label{alg:line:update if-4}
	    }
	    
	    \If{$ w(S^- \rightarrow P_i) < w(S^- \rightarrow P_j)$ for some $j \neq i$}
	    { 
	        \label{alg:line:condition if-5}
	        Let $P_i = P_i \setminus S^-$, merge $S^-$ with $\mathrm{argmax}_{P_j} \{ w(S^- \rightarrow P_j)\}$, and \textbf{go to} Line~\ref{alg:while line};\label{alg:line:update if-5}
	    }
	}

	For every partition set $P_i$  with \textsf{HC} tree $\T_i$, let $\mathcal{S}_i$ be the set of critical nodes;
	\label{alg:line:compute trees}
	
	\If {$\exists N \in \mathcal{S}_i$ such that $\max \left\{ \Phi_G(N^+), \Phi_G\left(\barr{N^+}\right) \right \}\leq \lp 1 + \frac{1}{k+1} \rp^{r+1} \cdot \rho^*$}
	{
	    \label{alg:line:condition if-6}
	    Let $P_i = P_i \setminus \overline{N^+}, \core(P_i) = N^+$, $P_{r+1} = \core(P_{r+1}) = \barr{N^+}$,   $r = r + 1$ and \textbf{go to} Line~\ref{alg:while line}; \label{alg:line:update if-6}
	}
	\If {$\exists N \in \mathcal{S}_i$ such that $\vol(N^+) \leq \vol(\core(P_i)) /2$ and $\varphi(N^+, \core(P_i)) \leq \frac{1}{3(k+1)}$} 
	{
	    \label{alg:line:condition if-7}
	    Update $\core(P_i)$ to either $N^+$ or $\barr{N^+}$ with the lower conductance, and \textbf{go to} Line~\ref{alg:while line}; \label{alg:line:update if-7}
	}
	\If {$\exists N \in \mathcal{S}_i$ such that $\vol(N^-) \leq \vol(P_i)/2$ and $ w(N^- \rightarrow P_i) < w(N^- \rightarrow P_j)$ for some $j \neq i$}
	{
	    \label{alg:line:condition if-8}
	    Let $P_i = P_i \setminus N^-$, merge $N^-$ with $\mathrm{argmax}_{P_j} \{ w(N^- \rightarrow P_j)\}$, and \textbf{go to} Line~\ref{alg:while line};\label{alg:line:update if-8}
	}
	
	\Return the partition $\{P_i\}_{i=1}^r$.
	
	\caption{Algorithm for partitioning $G$ into $\ell \leq k$ clusters }
	\label{alg:Strong Decomposition}
\end{algorithm}
\footnotetext{If the set $S \subset P_i$ returned by \SpecPart has $\vol(S^+) > \vol(G)/2$, swap $S$ with $P_i\setminus S$.}

Following the proof structure in \cite{GT14},  we will  prove Lemma~\ref{lem:Improved Decomposition} via a sequence of claims.  Notice that, during the  the entire execution  of the algorithm, the sets $\{P_i\}_{i=1}^{r}$ always form a partition of $V$,  and each $\core(P_i)$  is a subset of $P_i$. Firstly, we show that, at any point during the execution of the algorithm, the core sets $\core(P_i)~(1\leq i\leq r)$ always have low conductance.

\begin{claim}\label{claim:Strong decomp claim 1}
    Throughout the algorithm, we always have that 
    \[
        \max_{1 \leq i \leq r} \Phi_G(\core(P_i)) \leq \rho^* \cdot \lp 1 + \frac{1}{k+1} \rp ^{r}.
    \]
\end{claim}

The following result will be used in our proof:
\begin{lemma}[Lemma~2.2,  \cite{GT14}]\label{lem:Refinement B_i}
Let $G=(V,E, w)$ be a graph, and let $S, W$   be two subsets such that $S \subset W \subseteq V$.  
 Suppose that the following two conditions are satisfied for some $\varepsilon > 0$:
\begin{enumerate}
    \item $\varphi(S, W) \leq  \varepsilon /3$ and 
    \item $\max \left\{ \Phi_G(S), \Phi_G(W \setminus S)\right\} \geq (1 + \varepsilon)\cdot \Phi_G(W)$.
\end{enumerate}
Then it holds that 
\[
    \min \left\{ \Phi_G(S), \Phi_G(W \setminus S)\right\} \leq \Phi_G(W).
\]
\end{lemma}

\begin{proof}[Proof of Claim~\ref{claim:Strong decomp claim 1}]

Let $r$ be the current number of clusters generated by the algorithm, and we prove by induction that the claim holds during the entire execution of the algorithm. First of all, for  the base case of  $r = 1$, we have that $\core(P_1) = P_1 = V$, which gives us that $\Phi_G(\core(P_1))=0$; hence,  the statement holds trivially.

Secondly, for the inductive step, we assume that the statement holds for some fixed configuration of the core sets $\{\mathrm{core}(P_i)\}_{i=1}^r$ and we prove that the statement holds after the algorithm updates the current configuration. Notice that   $\{\core(P_i)\}_{i=1}^r$ are updated through Lines~\ref{alg:line:update if-1}, \ref{alg:line:update if-2}, \ref{alg:line:update if-3}, \ref{alg:line:update if-6} and \ref{alg:line:update if-7} of the algorithm, so it suffices to show that the claim holds after executing  these lines. We continue the proof with case distinction.
\begin{itemize}
    \item  When executing Lines~\ref{alg:line:update if-1}, \ref{alg:line:update if-3} and \ref{alg:line:update if-6}, the algorithm introduces some new set $\core(P_{r+1})$ such that
    \[
    \Phi_G(\core(P_{r+1})) \leq \rho^*\cdot \left(1 +\frac{1}{k+1} \right)^{r+1}.
    \]
Combining this with the inductive hypothesis, which assumes the inequality holds for $\core(P_{i})~(1\leq i\leq r)$, we have that 
\[
\max_{1\leq i\leq r+1}\Phi_G(\core(P_{i})) \leq \rho^*\cdot \left(1 +\frac{1}{k+1} \right)^{r+1}.
\]
\item The case for executing Lines~\ref{alg:line:update if-2}~and~\ref{alg:line:update if-7} is similar, so we focus on Line~\ref{alg:line:update if-2} here and prove this by applying Lemma~\ref{lem:Refinement B_i}. When executing Line~\ref{alg:line:update if-2}, we know that the \texttt{if}-condition in Line~\ref{alg:line:condition if-1} does not hold, so we have that 
    \[
        \max \left\{ \Phi_G(S^+), \Phi_G\left(\barr{S^+}\right) \right\} > \lp 1 + \frac{1}{k+1}\rp^{r +1} \cdot \rho^* \geq \lp 1 + \frac{1}{k+1}\rp \cdot \Phi_G(\core(P_i)),
    \]
    where the last inequality follows by the inductive hypothesis. Moreover, when executing Line~\ref{alg:line:update if-2}, we also know that the \texttt{if}-condition in Line~\ref{alg:line:condition if-2}  holds, i.e.,
    \[
        \min \left\{ \varphi\left(S^+, \core(P_{i})\right), \varphi\left(\barr{S^+}, \core(P_{i})\right)\right\} \leq \frac{1}{3(k+1)}.
    \]
    Therefore, by applying Lemma~\ref{lem:Refinement B_i} with  $S^+ \subset \mathrm{core}(P_i)$ and $\varepsilon = 1/(k+1)$ and using the inductive hypothesis, we conclude that
    \[
        \min \left\{ \Phi_G(S^+), \Phi_G\left(\barr{S^+}\right)\right\} \leq \Phi_G(\core(P_{i})) \leq\rho^*\cdot  \lp 1 + \frac{1}{k+1} \rp^{r}.
    \]
\end{itemize}
Combining the two cases above, we know that the claim always holds during the entire execution of the algorithm. This completes the proof.
\end{proof}

Next, we will show that the number of partition sets cannot exceed $k$. This proof  is identical to Claim~$3.2$ in \cite{GT14}, and we include the proof  here for completeness.  

\begin{claim}\label{claim:Strong decomp claim 2}
    The total number of clusters returned by the algorithm satisfies that $\ell\leq k$.
\end{claim}

\begin{proof}

Suppose for contradiction that the number of clusters becomes $r = k + 1$ at some point during the execution of the algorithm. Then, since $\core(P_1),\ldots, \core(P_{k+1})$ are disjoint, by the definition of $\rho(k+1)$ and Claim~\ref{claim:Strong decomp claim 1} we have that 
\[
        \rho(k+1) 
        \leq \max_{1 \leq i \leq k+1} \Phi_G(\mathrm{core}(P_i))
        \leq \lp 1 + \frac{1}{k+1}\rp^{k+1} \cdot \rho^*
        \leq \mathrm{e} \cdot \rho^*
        \leq \mathrm{e} \cdot \frac{\lambda_{k+1}}{10}
        < \frac{\lambda_{k+1}}{2},
    \]
    which contradicts Lemma~\ref{lem:Higher Cheeger}.  Therefore, the total number of clusters at any time satisfies $r < k+1$, which means that the final number of clusters satisfies $\ell \leq k$. This proves the claim.
\end{proof}

Now we are ready to show   that the output  $\{P_i\}_{i=1}^{\ell}$ of Algorithm~\ref{alg:Strong Decomposition} and its core sets $\{\core(P_i) \}_{i=1}^{\ell}$ satisfy Properties~$(A4)$ and $(A5)$, which will be used in proving Lemma~\ref{lem:Improved Decomposition}. 

\begin{claim}\label{claim:Strond decomp claim 3}
Let $\{P_i\}_{i=1}^{\ell}$ be the output of Algorithm~\ref{alg:Strong Decomposition} with the corresponding core sets $\{\core(P_i)\}_{i=1}^{\ell}$.
Then, the following hold for any $1\leq i\leq \ell$:
\begin{enumerate}
    \item $\Phi_G(\core(P_i))\leq \phiOut/(k+1)$;
    \item $\Phi_G(P_i) \leq \phiOut$;
    \item $\Phi_{G[P_i]}\geq \phiIn^2/4 $.
\end{enumerate}
Moreover, assuming that $\mathcal{S}_i$ is the set of critical nodes of $\T_i = \T_{\deg}(G[P_i])$, the following two properties hold for any $1\leq i\leq \ell$:
\begin{enumerate}
    \item[4.] For every critical node $N\in \mathcal{S}_i$ with $\vol(N^+)\leq \vol(\core(P_i))/2$, we have that
    \[
        \varphi(N^+, \core(P_i)) \geq \frac{1}{3(k+1)};
    \]
    \item[5.] 
    For every critical node $N\in \mathcal{S}_i$ with
     $\vol(N^-) \leq \vol(P_i)/2$, we have  that $$w(N^- \rightarrow P_i) \geq w(N^- \rightarrow V \setminus P_i) \cdot \frac{1}{k+1}.$$
\end{enumerate}
\end{claim}

\begin{proof}
First of all, by Claim~\ref{claim:Strong decomp claim 1} we have for any $1\leq i\leq \ell$ that 
 \[
    \Phi_G(\core(P_i)) 
    \leq \rho^* \cdot \lp 1 + \frac{1}{k+1}\rp^{\ell}
    \leq \mathrm{e} \cdot \rho^*
    \leq 30 \cdot \mathrm{e} \cdot c_0 \cdot (k+1)^5 \sqrt{\lambda_k}
    \leq \frac{\phiOut}{k+1},
\]
where the second inequality holds by the fact  that $\ell\leq k+1$, the third one holds by the choice of $\rho^*$, and the last one holds by the choice of $\phiOut$. This proves Item~(1).

To prove Item~(2), we notice that the first condition of the \texttt{while}-loop~(Line~\ref{alg:while line}) doesn't hold when the algorithm terminates, hence   we have for any $1\leq i\neq j\leq \ell$ that \[w(P_i \setminus\core(P_i)\rightarrow P_i)\geq w(P_i \setminus \core(P_i) \rightarrow P_j).
\]
By applying the averaging argument, we have that 
\begin{align}
    w(P_i \setminus \core(P_i) \rightarrow \core(P_i)) & = w(P_i \setminus \core(P_i) \rightarrow P_i) \nonumber\\
    &\geq \frac{w(P_i \setminus \core(P_i) \rightarrow V)}{\ell}
    \geq \frac{w(P_i \setminus \core(P_i) \rightarrow V \setminus P_i)}{k}. \label{eq:claim Strong Decomp Claim 3}
\end{align} 
We apply the same analysis used in \cite{GT14}, and have that
\begin{align*}
    \Phi_G(P_i) 
    &= \frac{w(P_i \rightarrow V)}{\vol(P_i)} \\
    &\leq \frac{w \lp \mathrm{core}(P_i) \rightarrow V \rp + w \lp P_i \setminus \mathrm{core}(P_i)\rightarrow V \setminus P_i \rp - w(P_i \setminus \mathrm{core}(P_i) \rightarrow \mathrm{core}(P_i)) }{\vol(\mathrm{core}(P_i))}\\
    &\leq \Phi_G(\mathrm{core}(P_i)) + \frac{(k-1) \cdot w(P_i \setminus \mathrm{core}(P_i) \rightarrow \mathrm{core}(P_i))}{\vol(\mathrm{core}(P_i))}\\
    &\leq k \cdot \Phi_G(\mathrm{core}(P_i))\\
    &\leq \phiOut,
\end{align*}
where the second inequality uses equation~\eqref{eq:claim Strong Decomp Claim 3}. This proves Item~(2).

Next, we analyse Item~(3). Again, we know that   the second condition within the \texttt{while}-loop~(Line~\ref{alg:while line}) does not hold when the algorithm terminates. By the performance of  the \SpecPart algorithm~(i.e., Lemma~\ref{lem:Cheeger's ineq}), it holds for any $1\leq i\leq \ell$ that 
$\Phi_{G[P_i]}\geq  \phiIn^2/4$. With this, we prove that   Item~(3) holds.
    
 Similarly, when the algorithm terminates, we know that for any critical node $N$ the \texttt{if}-condition in Line~\ref{alg:line:condition if-7} does not hold.
Hence, for any $1\leq i\leq \ell$ and any $N\in \mathcal{S}_i$ with $\vol(N^+)\leq \vol(\core(P_i))/2$, we have that
\[
    \varphi(N^+, \core(P_i)) \geq \frac{1}{3(k+1)}.
\]
This shows that Item~(4) holds as well.
    
Finally, since there is no $N$ satisfying the \texttt{if}-condition in Line~\ref{alg:line:condition if-8} of the algorithm, it holds for any $1\leq i\neq j\leq \ell$ and every critical node $N\in\mathcal{S}_i$ that $w(N^{-} \rightarrow P_i) \geq w(N^{-} \rightarrow P_j)$. Therefore, by the same averaging argument we have that
\[
    w(N^- \rightarrow P_i) 
    \geq \frac{w(N^- \rightarrow V)}{\ell}
    \geq \frac{w(N^- \rightarrow V\setminus P_i)}{k+1},
\]
which shows that Item~(5) holds.
\end{proof}

It remains to prove that the algorithm does terminate. To prove this, we first show that, in each iteration of the \texttt{while}-loop (Lines~\ref{alg:line:beginning while loop}--\ref{alg:line:update if-5}), at least one of the \texttt{if}-conditions will be satisfied, and some sets are updated accordingly. This fact, stated as  Claim~\ref{claim:Strong decomp claim 4}, is important, since otherwise the algorithm might end up
in   an infinite loop. The following result will be used in our proof.

\begin{lemma}[Lemma~2.6, \cite{GT14}] \label{lem:Inner Conductance P_i}
    Let $\mathrm{core}(P_i) \subseteq P_i \subseteq V$, and  $S \subset P_i$ be such that $\vol(S^+)~\leq~\vol(\mathrm{core}(P_i))/2$.  Suppose that the following hold for some parameters $\rho$ and $0 < \varepsilon < 1$:
    \begin{enumerate}
        \item $\rho \leq \Phi_G(S^-)$ and $\rho \leq \max \{ \Phi_G(S^+), \Phi_G(\barr{S^+})\}$;
        \item  If $S^- \neq \emptyset$, then $w(S^- \rightarrow P_i) \geq w(S^- \rightarrow V)/k$;
        \item If $S^+ \neq \emptyset$, then $\varphi(S^+, \mathrm{core}(P_i)) \geq \varepsilon/3$ and $\varphi\left(\barr{S^+}, \mathrm{core}(P_i)\right) \geq \varepsilon/3$.
    \end{enumerate}
    Then,  it holds that $$\Phi_{G[P_i]}(S) \geq \varepsilon \cdot \frac{\rho}{14k}.$$
\end{lemma}

\begin{claim}\label{claim:Strong decomp claim 4}
    If at least one condition of the \texttt{while}-loop is satisfied, then at least one of the \texttt{if}-conditions~(Lines~\ref{alg:line:condition if-1},\ref{alg:line:condition if-2},\ref{alg:line:condition if-3},\ref{alg:line:condition if-4} or \ref{alg:line:condition if-5}) is satisfied.  
\end{claim}

\begin{proof}
First of all, notice that if the first condition of the \texttt{while}-loop  is satisfied, then the \texttt{if}-condition in Line~\ref{alg:line:condition if-4} will be satisfied and the claim holds. Hence,   we assume   that only the second condition of the \texttt{while}-loop is satisfied,  and we   prove the claim by contradiction. That is, we   show that, if none of the \texttt{if}-conditions holds, then the set $S$ returned by  the \SpecPart algorithm would satisfy that $\Phi_{G[P_i]}(S)\geq \phiIn$.
 The proof is structured in the following two steps:
\begin{enumerate}
    \item We  first prove that  $\Phi_{G[P_i]}(S) \geq \frac{\max \{ \rho^*, \rho(r+1)\}}{14(k+1)^2}$;
    \item Using Item~$(1)$ we prove that $\Phi_{G[P_i]}(S)\geq \phiIn$ and reach our desired contradiction.
\end{enumerate} 

 \paragraph{Step~1:}
    We prove this fact by applying Lemma~\ref{lem:Inner Conductance P_i} with parameters
    \[
        \rho \triangleq \max \{\rho^*, \rho(r+1)\}
        \quad \text{and}
        \quad \varepsilon \triangleq \frac{1}{k+1}.
    \]
    Let us show that the conditions of Lemma~\ref{lem:Inner Conductance P_i} are satisfied, beginning with the first one. If $S^{-} = \emptyset$, then we trivially have that $1 = \Phi_G(S^{-}) \geq \rho$; so we assume that $S^{-} \neq \emptyset$. As the \texttt{if}-condition in Line~\ref{alg:line:condition if-3} is not  satisfied, we have that 
    \[
        \Phi_G(S^-)\geq \left(1 + \frac{1}{k+1} \right)^{r+1}\cdot \rho^*,
    \]
    and combining this with Claim~\ref{claim:Strong decomp claim 1} gives us that
    \[
        \Phi_G(S^-) = \max 
        \left\{ \Phi_G(\core(P_1)), \dots, \Phi_G(\core(P_r)), \Phi_G(S^-)\right\} \geq \rho(r+1).
    \]
    Therefore, we have that 
    \begin{equation}\label{eq:Lower Bound Cond S-}
        \Phi_G(S^-) \geq \max\{\rho^*, \rho(r+1)\} = \rho.
    \end{equation}
    Similarly, if $S^+ = \emptyset$, then we trivially have that $1 = \max \{\Phi_G(S^+), \Phi_G(\overline{S^+})\} \geq \rho$; so we assume that $S^+ \neq \emptyset$. Moreover, since
    we have chosen $S^+$ such that $\vol(S^+) \leq \vol(\mathrm{core}(P_i))/2$, we know that $\overline{S^+} = \mathrm{core}(P_i) \setminus S^+ \neq \emptyset$. As the \texttt{if}-condition in Line~\ref{alg:line:condition if-1} is not satisfied, we have that
    \[
    \max\left\{\Phi_G(S^+), \Phi_G\left(\overline{S^+}\right)\right\}\geq \left(1 + \frac{1}{k+1}\right)^{r+1} \cdot\rho^*.
    \]
    Combining this with Claim~\ref{claim:Strong decomp claim 1} gives us that
    \begin{align*}
        \lefteqn{\max\left\{\Phi_G(S^+), \Phi_G\left(\overline{S^+}\right)\right\}}\\
        &=  \max \left\{ \Phi_G(\core(P_1)), \dots, \Phi_G(\core(P_{i-1})), \Phi_G(S^+), \Phi_G\left(\barr{S^+}\right), \Phi_G(\core(P_{i+1})), \dots, \Phi_G(\core(P_{r}))\right\}\\
        &\geq   \rho(r+1).
    \end{align*}
    Therefore we have that 
    \begin{equation}\label{eq:Lower bound S+ and S+bar}
     \max\left\{\Phi_G(S^+), \Phi_G\left(\overline{S^+}\right)\right\} \geq \max\{\rho^*, \rho(r+1)\} = \rho.
    \end{equation}
    Combining  \eqref{eq:Lower Bound Cond S-} and \eqref{eq:Lower bound S+ and S+bar},  we see that the first condition of Lemma~\ref{lem:Inner Conductance P_i} is satisfied.
    Since the \texttt{if}-condition in Line~\ref{alg:line:condition if-5} is not satisfied, it follows by an averaging argument that 
    \[
         w(S^- \rightarrow P_i) \geq  \frac{w(S^- \rightarrow V)}{k}, 
    \]
    which shows that the second condition of Lemma~\ref{lem:Inner Conductance P_i} is satisfied.
    Finally, since the \texttt{if}-condition in Line~\ref{alg:line:condition if-2} is not satisfied, we know that
    \[
        \min \left\{ \varphi(S^+, \core(P_i)), \varphi\left(\barr{S^+}, \core(P_i)\right) \right\} \geq \frac{1}{3(k+1)},
    \]
    which shows that the third condition of Lemma~\ref{lem:Inner Conductance P_i} is satisfied as well. 
Hence, by Lemma~\ref{lem:Inner Conductance P_i} we conclude that
\begin{equation}\label{eq:claim:Strong decomp claim 4 eq1}
    \Phi_{G[P_i]}(S) 
    \geq \frac{\varepsilon \cdot \rho}{14(k+1)}
    = \frac{\max\{ \rho^*, \rho(r+1)\} }{14(k+1)^2},
\end{equation}
which completes the proof of the first step.

\paragraph{Step~2:}We prove this step with a case distinction as follows.

    \emph{Case~1: $r = k$.} By \eqref{eq:claim:Strong decomp claim 4 eq1} and Lemma~\ref{lem:Higher Cheeger}, we have that 
    \[
        \Phi_{G[P_i]}(S) 
        \geq \frac{\rho(r+1)}{14(k+1)^2}
        = \frac{\rho(k+1)}{14(k+1)^2}
        \geq \frac{\lambda_{k+1}}{28(k+1)^2}
        \geq \phiIn,
    \]
    which leads to the desired contradiction.  
    
    \emph{Case~2: $r < k$.} Recall that the  partition  sets $\{P_i\}_{i=1}^{r}$ are labelled such that $\lambda_2(G[P_1]) \leq \ldots \leq \lambda_2(G[P_r])$, and the algorithm has chosen the lowest index $i$ for which the set $S\subset P_i$ returned by the \SpecPart algorithm satisfies the second condition of the \texttt{while}-loop.
     Our proof is based on  a further case distinction depending  on the value of  $i$.

    \emph{Case~2a:} $i=1$~(i.e., the algorithm selects $S \subseteq P_1$).
    We combine  the performance of the \SpecPart  algorithm~(Lemma~\ref{lem:Cheeger's ineq}) with  Lemma~\ref{lem:Upperbound eig induced graphs}, and obtain that
    \begin{equation}\label{eq:claim:Strong decomp claim 4 eq2}
        \Phi_{G[P_1]}(S) \leq  \sqrt{2\lambda_2(G[P_1])} = \min_{1 \leq j \leq r} \sqrt{2 \lambda_2(G[P_j])} \leq \sqrt{4c_0 \cdot k^6 \cdot \lambda_k}.
    \end{equation}
    Combining   \eqref{eq:claim:Strong decomp claim 4 eq1} and \eqref{eq:claim:Strong decomp claim 4 eq2} we have that 
    \[
        \rho^* \leq 28 c_0 \cdot (k+1)^5\sqrt{\lambda_k}.
    \]
    Thus, by the definition of $\rho^*$ we   have that 
    \[
        \rho^* = \frac{\lambda_{k+1}}{10}.
    \]
    We combine this with  \eqref{eq:claim:Strong decomp claim 4 eq1}, and have that 
    \[
        \Phi_{G[P_i]}(S) \geq \frac{\lambda_{k+1}}{140 (k+1)^2} = \phiIn,
    \]
    which gives our desired contradiction.
    
    \emph{Case~2b: $i>1$}~(i.e., the algorithm selects $S \subset P_i$  for some $i \geq 2$).
    Let $S_1 \subset P_1$ be the set obtained by applying the \SpecPart algorithm to the graph $G[P_1]$. Since the algorithm did not select $S_1 \subset P_1$, we know that $\Phi_{G[P_1]}(S_1) \geq \phiIn$. Combining the performance of the \SpecPart algorithm~(Lemma~\ref{lem:Cheeger's ineq}) with Lemma~\ref{lem:Upperbound eig induced graphs}, we have that  
    \[
        \phiIn \leq \Phi_{G[P_1]}(S_1) \leq \min_{1 \leq j \leq r} \sqrt{2 \lambda_{2}(G[P_j])} \leq 2c_0 \cdot k^3 \cdot \sqrt{\lambda_k}.
    \]
       This gives us  that
    \[
         \frac{\lambda_{k+1}}{10} = 14(k+1)^2 \cdot \phiIn < 30c_0 \cdot (k+1)^5 \cdot \sqrt{\lambda_k},
    \]
    and it holds by the definition of $\rho^*$ that
    \[
        \rho^* = \frac{\lambda_{k+1}}{10}.
    \]
    Therefore,  by \eqref{eq:claim:Strong decomp claim 4 eq1} we have that
    \[
        \Phi_{G[P_i]}(S) \geq \frac{\lambda_{k+1}}{140 (k+1)^2} = \phiIn.
    \]
    Combining the two cases above gives us the desired contradiction. With this, we complete the proof of the claim. 
\end{proof}

Next, we will show that the total number of iterations that the algorithm runs, i.e.,  the number of times the instruction ``\texttt{go to Line~\ref{alg:while line}}'' is executed, is finite.

\begin{claim}\label{claim:Strong decomp claim 5}
For any graph $G=(V,E,w)$ with the minimum weight $w_{\min}$ as the input, Algorithm~\ref{alg:Strong Decomposition} terminates after executing the while-loop $O \lp  k \cdot n \cdot \vol(G)/w_{\min}\rp$ times.
\end{claim}

\begin{proof} 
Notice that the algorithm goes back to check the loop conditions~(Line~\ref{alg:while line})
right after any of   Lines~\ref{alg:line:update if-1}, \ref{alg:line:update if-2}, \ref{alg:line:update if-3}, \ref{alg:line:update if-4}, \ref{alg:line:update if-5}, \ref{alg:line:update if-6}, \ref{alg:line:update if-7}~and~\ref{alg:line:update if-8} is executed, and each of these commands changes the current structure of our partition $\{P_i\}_{i=1}^{r}$ with core sets $\core(P_i) \subseteq P_i$. We classify these updates into the following three types:
    \begin{enumerate}
        \item The updates that introduce  a new partition set $P_{r+1}$. These correspond to Lines~\ref{alg:line:update if-1}, \ref{alg:line:update if-3} and \ref{alg:line:update if-6};
        \item The updates  that contract the core sets $\core(P_i)$ to a strictly smaller subset $T \subset \core(P_i)$.  These correspond to Lines~\ref{alg:line:update if-2} and \ref{alg:line:update if-7}; 
        \item The updates that refine the partition sets $\{P_i\}_{i=1}^{r}$ by moving  a subset $ T \subseteq P_i \setminus \core(P_i)$ from the partition set $P_i$ to a different partition set $P_j$, for some $P_i \neq P_j$. These correspond to Lines~\ref{alg:line:update if-4}, \ref{alg:line:update if-5}, and \ref{alg:line:update if-8}. 
    \end{enumerate}
    We prove that these updates can occur only a finite number of times. The first type of updates can occur at most $k$ times, since we know by Claim~\ref{claim:Strong decomp claim 2}  that the algorithm outputs $\ell\leq k$ clusters.  Secondly, for a fixed value of $\ell$, the second type of updates occurs  at most $n$ times, since each update decreases the size of some $\core(P_i)$ by at least one. Finally, for a fixed $\ell$ and a fixed configuration of core sets $\core(P_i) \subseteq P_i$, the third type of updates occurs at most  $O(\vol(G)/w_{\min})$  times. This is due to the fact that, whenever every such update is  executed, the total weight between different partition sets, i.e., $\sum_{i\neq j } w(P_i, P_j)$, decreases by at least $w_{\min}$.   Combining everything together proves the lemma.
\end{proof}

Finally, we bring everything together and prove Lemma~\ref{lem:Improved Decomposition}.

\begin{proof}[Proof of Lemma~\ref{lem:Improved Decomposition}]
We first show that Properties~$(A1), (A2)$ and $(A3)$ hold, and in the end we analyse the runtime of the algorithm. Combining Items~(2) and (3)  of Claim~\ref{claim:Strond decomp claim 3} with the choices of $\phiIn, \phiOut$ in \eqref{eq:PhiIn PhiOut Choices}, we obtain for all $1 \leq i \leq \ell$ that $\Phi_G(P_i) \leq \phiOut = O \lp k^6 \cdot \sqrt{\lambda_k}\rp$ and $\Phi_{G[P_i]} \geq  \phiIn^2/ 4 = \Omega \lp\lambda_{k+1}^2/k^4 \rp$. Hence, Properties $(A1)$ and $(A2)$ hold for every $P_i$.

To analyse Property~$(A3)$, we fix an arbitrary node $N\in\mathcal{S}_i$ that belongs to the partition set $P_i$ with core set $\core(P_i)$.
By definition, we have that 
\[
    w(N, V\setminus P_i) = w(N^+, V\setminus P_i) + w(N^-, V\setminus P_i).
\] 
We study $w(N^+, V\setminus P_i)$ and $ w(N^-, V\setminus P_i)$ separately.

\paragraph{Bounding  the value of  $w(N^+, V \setminus P_i)$:} We analyse $w(N^+, V \setminus P_i)$ by the following case distinction.

\underline{Case~1: $\vol(N^+) \leq \vol(\core(P_i))/2$.}
By Item~(4)  of Claim~\ref{claim:Strond decomp claim 3} we know that 
\[
    \varphi(N^+, \core(P_i)) \geq \frac{1}{3(k+1)},
\]
which is equivalent to 
\[
    3(k+1) \cdot \frac{\vol(\core(P_i))}{\vol(\barr{N^+})} \cdot w(N^+ \rightarrow \core(P_i)) \geq  w(N^+ \rightarrow V \setminus \core(P_i)).
\]
This implies that 
\[
    6(k+1) \cdot w(N^+ \rightarrow \core(P_i)) \geq w(N^+ \rightarrow V \setminus \core(P_i)),
\]
and we have   that
    \begin{align*}
        w(N^+, V \setminus P_i) 
        & 
        \leq w(N^+ \rightarrow V \setminus \core(P_i)) 
         \leq 6(k+1) \cdot w(N^+ \rightarrow \core(P_i))\\
         & 
        \leq 6(k+1) \cdot \vol_{G[P_i]}(N^+).
    \end{align*}
    
    \underline{Case~2: $\vol(N^+) >  \vol(\core(P_i))/2$.}
    We have that 
    \begin{align}
        w(N^+, V \setminus P_i) 
        &\leq w(\core(P_i), V \setminus P_i) 
        \leq w(\core(P_i), V \setminus \core(P_i)) \nonumber
       \\
       & = \vol(\core(P_i)) \cdot \Phi_G(\core(P_i))\leq \vol(\core(P_i)) \cdot \frac{\phiOut}{k+1} \nonumber \\
       & 
        < \frac{2\phiOut}{k+1} \cdot \vol(N^+) \label{eq:claim:cost crossing edges eq1}, 
    \end{align}
     where the third  inequality follows by Item~(1)  of Claim~\ref{claim:Strond decomp claim 3}. Therefore,  we have that 
    \begin{equation}\label{eq:claim:cost crossing edges eq2}
        \vol_{G[P_i]}(N^+) = \vol(N^+) - w(N^+, V \setminus P_i) > \vol(N^+) \lp 1 - \frac{2\phiOut}{k+1} \rp,
    \end{equation}
    where the last inequality follows by \eqref{eq:claim:cost crossing edges eq1}. We further combine   \eqref{eq:claim:cost crossing edges eq1} with  \eqref{eq:claim:cost crossing edges eq2}, and obtain  that 
    \[
        w(N^+, V \setminus P_i) 
        \leq \frac{2\phiOut}{k+1} \cdot \frac{1}{1 - \frac{2\phiOut}{k+1}} \cdot \vol_{G[P_i]}(N^+)
        \leq \frac{2\phiOut}{k} \cdot \vol_{G[P_i]}(N^+),
    \]
    where the last inequality holds by our assumption that $\phiOut < 1/3$.
    
    Therefore,  combining the two cases above gives us that 
    \begin{equation}\label{eq:claim:cost crossing edges eq3}
        w(N^+, V \setminus P_i) \leq   6(k+1)   \cdot \vol_{G[P_i]}\left(N^+\right).
    \end{equation}  
    
    \paragraph{Bounding the value of  $w(N^-, V \setminus P_i)$:}
    We analyse   $w(N^-, V \setminus P_i)$ based on the following two cases. 
    
    \underline{Case~1: $\vol(N^-) \leq \vol(P_i)/ 2$.} By Item~(5) of Claim~\ref{claim:Strond decomp claim 3}, we know that 
    \[
        w(N^- \rightarrow P_i) \geq w(N^-\rightarrow V \setminus P_i) \cdot \frac{1}{(k+1)},
    \]
    which gives us that   
    \[
        w(N^-, V \setminus P_i) \leq (k+1) \cdot w(N^- \rightarrow P_i) \leq (k+1)
        \cdot \vol_{G[P_i]}(N^-).
    \]
    
   \underline{Case 2: $\vol(N^-) >  \vol(P_i)/2$.} In this case, we have that
    \begin{equation}\label{eq:claim:cost crossing edges eq4}
        w(N^-, V \setminus P_i) 
        \leq w(P_i, V \setminus P_i) 
         = \Phi_G(P_i) \cdot \vol(P_i) 
        \leq \phiOut \cdot \vol(P_i)
        \leq 2\phiOut \cdot \vol(N^-),
    \end{equation}
    where the  second  inequality follows by  Item~(2) of Claim~\ref{claim:Strond decomp claim 3}. This implies that 
    \begin{equation}\label{eq:claim:cost crossing edges eq5}
        \vol_{G[P_i]}(N^-) 
        = \vol(N^-) - w(N^-, V \setminus P_i) 
        \geq (1 - 2\phiOut) \cdot \vol(N^-),
    \end{equation}
    where the last inequality follows by  \eqref{eq:claim:cost crossing edges eq4}. Finally,  combining \eqref{eq:claim:cost crossing edges eq4} and \eqref{eq:claim:cost crossing edges eq5} gives us that
    \[
        w(N^-, V \setminus P_i) \leq \frac{2\phiOut}{1 - 2\phiOut} \cdot \vol_{G[P_i]}(N^-) \leq 2 \cdot \vol_{G[P_i]}(N^-),
    \]
    where the last inequality follows by our assumption that $\phiOut < 1/3$. 
    Therefore,   combining the two cases together gives us that  
    \begin{equation}\label{eq:claim:cost crossing edges eq6}
        w(N^-, V \setminus P_i) \leq (k+1) \cdot \vol_{G[P_i]}(N^-).
    \end{equation}
    
    Our claimed property $(A3)$ follows by summing the inequalities in   \eqref{eq:claim:cost crossing edges eq3} and \eqref{eq:claim:cost crossing edges eq6} and the fact that  $\vol_{G[P_i]}(N) = \vol_{G[P_i]}(N^+) + \vol_{G[P_i]}(N^-)$.
    
Finally, we analyse the runtime of the algorithm.
By   Claims~\ref{claim:Strong decomp claim 4} and \ref{claim:Strong decomp claim 5}, we know that the algorithm does terminate, and the total number of iterations of the main \texttt{while}-loop executed by the algorithm is upper bounded by  $O(k\cdot n \cdot\vol(G) / w_{\min})$. Notice that   this quantity is upper bounded by $O(\poly(n))$ given our assumption that $w_{\max}/w_{\min} = O(\poly(n))$.  This completes the proof. 
\end{proof}

\section{Experiments}\label{sec:Experiments}

We experimentally evaluate the performance of our proposed algorithm, and compare it against the three well-known linkage heuristics for computing hierarchical clustering trees, and 
different variants of the algorithm proposed in \cite{CAKMT17}, i.e.,~\texttt{Linkage++}, on both synthetic and real-world data sets. At a high level, \texttt{Linkage++} consists of the following three steps: 
\begin{itemize}
    \item[(i)] Project the input data points into a lower dimensional Euclidean subspace; 
    \item[(ii)] Run the \texttt{Single Linkage} algorithm \cite{cohen2018hierarchical} until $k$ clusters are left; 
    \item[(iii)] Run a \texttt{Density} based linkage algorithm on the $k$ clusters until one cluster is left.
\end{itemize}
  Specifically, our algorithm \texttt{PruneMerge} will be compared against the following $6$ algorithms:
\begin{itemize}
\item  \texttt{Average Linkage}, \texttt{Complete Linkage}, and \texttt{Single Linkage}: the three well-known linkage algorithms studied in the literature. We refer the reader to \cite{cohen2018hierarchical} for a complete description;

\item  \texttt{Linkage++}, \texttt{PCA+} and \texttt{Density}: the algorithm proposed in \cite{CAKMT17}, together with two variants also studied in \cite{CAKMT17}. The algorithm \texttt{PCA+} corresponds to running Steps~(i) and (ii) of \texttt{Linkage++} until one cluster is left (as opposed to $k$ clusters), while \texttt{Density} corresponds to running Steps~(i) and (iii) of \texttt{Linkage++}.
\end{itemize}

All algorithms were implemented in Python 3.8 and the experiments were performed using an Intel(R) Core(TM) i5-6500 CPU @ 3.20GHz processor, with 16 GB RAM. All of the reported costs below are averaged over $5$ independent runs. For the parameter $\phiIn$ used in the proof of  Lemma~\ref{lem:Improved Decomposition} (see 
\eqref{eq:PhiIn PhiOut Choices}), we set in our implementation $\phiIn = \max\left\{\frac{\lambda_{k+1}}{140\cdot (k+1)^2}, 2 \cdot \lambda_k\right\}$. This value is used throughout the experiments reported here.  Our code can be downloaded from
\begin{center}\url{https://github.com/bgmang/hierarchical-clustering-well-clustered-graphs.git}.
\end{center}

\subsection{Experiments on synthetic data sets} We first compare the performance of our algorithm with the aforementioned other algorithms on synthetic data sets.

 \paragraph{Clusters of the same size.}
Our first set of experiments employ input graphs generated according to random stochastic models, where all clusters have the same size. For our first experiment, we look at graphs generated  according to the standard Stochastic Block Model (\textsf{SBM}).
We first set the number of clusters as $k=3$, and the number of vertices in each cluster $\{P_i\}_{i=1}^3$ as $1,000$. We assume that any pair of vertices within each cluster is connected by an edge with probability $p$, and any pair of vertices from different clusters is connected by an edge with probability $q$. We fix the value $q=0.002$, and consider different values of $p\in[0.04, 0.2]$. Our experimental results are illustrated in Figure~\ref{Fig:SBM and HSBM}(a).

For our second experiment, we consider graphs generated according to a hierarchical stochastic block model \textsf{(HSBM)} \cite{cohen2018hierarchical}. This model assumes the existence of a ground-truth hierarchical structure of the clusters. For the specific choice of parameters, we set the number of clusters as $k = 5$, and the number of vertices in each cluster $\{P_i\}_{i=1}^5$ as $600$. For every pair of vertices 
$(u, v) \in P_i \times P_j$,  we assume that $u$ and $v$ are connected by an edge with probability $p$ if $i = j$;  otherwise $u$ and $v$ are connected by an edge with probability $q_{i, j}$ defined as follows: (i) for all $i \in \{1, 2, 3\}$ and $j \in \{4, 5\}, q_{i, j} = q_{j, i} = q_{\mathrm{min}}$; (ii) for $i \in \{1, 2\}$, $q_{i, 3} = q_{3, i} = 2 \cdot q_{\mathrm{min}}$; (iii) $q_{4, 5} = q_{5, 4} = 2 \cdot q_{\mathrm{min}}$; (iv) $q_{1, 2} = q_{2, 1} = 3 \cdot q_{\mathrm{min}}$. We fix the value $q_{\mathrm{min}} = 0.0005$ and consider different values of $p \in [0.04, 0.2]$. Our results are reported in  Figure~\ref{Fig:SBM and HSBM}(b). We remark that this choice of parameters resembles similarities with \cite{CAKMT17}, and this ensures that the underlying graphs exhibit a ground truth hierarchical structure of clusters.

As reported in Figure~\ref{Fig:SBM and HSBM}, our experimental results for both sets of graphs are similar, and the performance of our algorithm is marginally better than \texttt{Linkage++}.
This is well expected, as  \texttt{Linkage++} is specifically designed for the \textsf{HSBM}, in which all the clusters have the same inner density characterised by parameter $p$, and their algorithm achieves a $(1 +o(1))$-approximation for those instances.

\begin{figure}[h]
    \centering
    \hspace{-1cm}
    \begin{subfigure}[h]{.43\textwidth}
        \centering
        \scalebox{.25}{\includegraphics{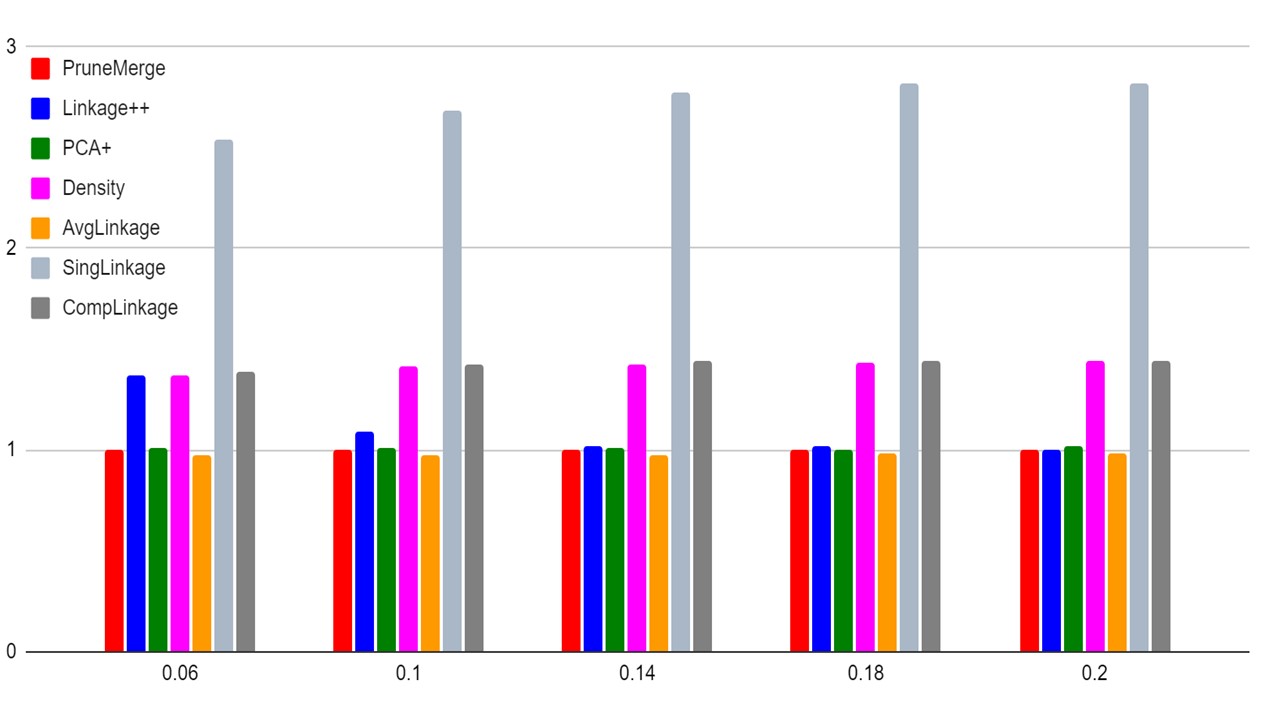}}
    \caption{}
    \label{Fig:SBM}
    \end{subfigure}
    \hspace{0.5cm}
    \begin{subfigure}[h]{.43\textwidth}
        \centering
        \scalebox{.25}{\includegraphics{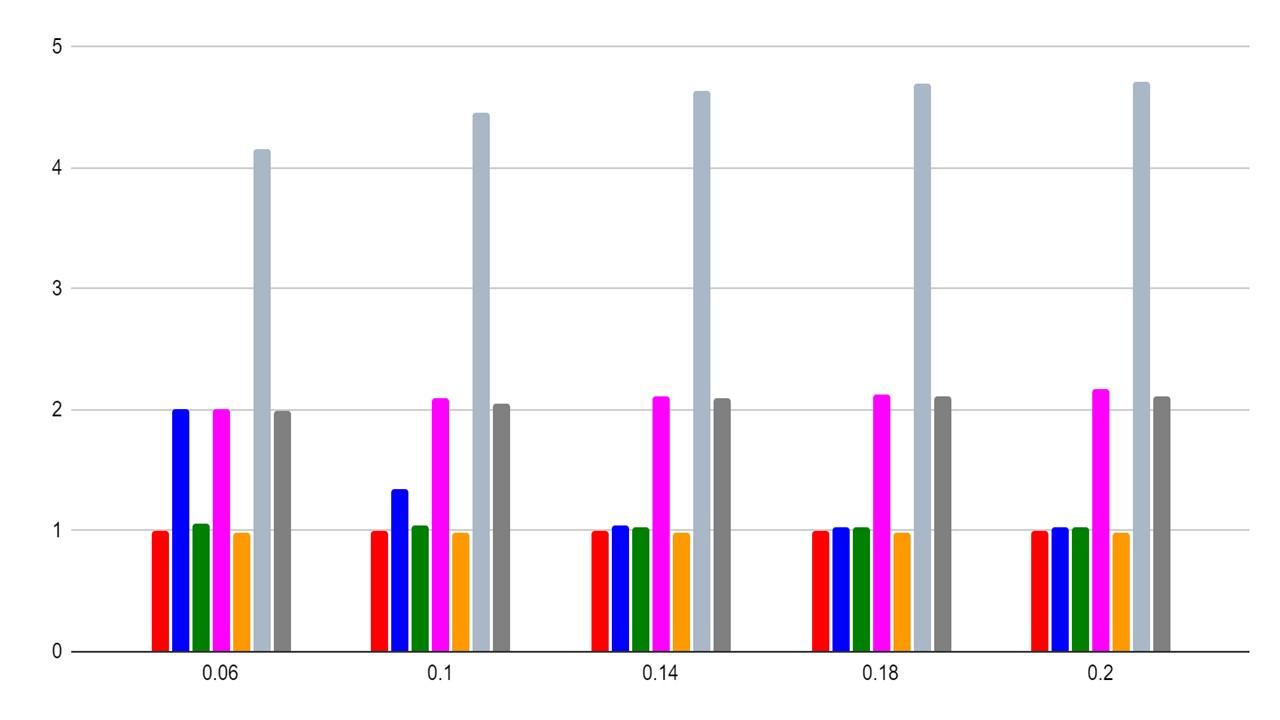}}
    \caption{}
    \label{Fig:HSBM}
    \end{subfigure}
    \caption{\small{Results for clusters of the same size. The $x$-axis represents different values of $p$, while the $y$-axis represents the cost of the algorithms' returned \textsf{HC} trees normalised by  
     the cost of \texttt{PruneMerge}. Figure~(a) corresponds to inputs generated according to the \textsf{SBM}, while Figure~(b) corresponds to those generated according to the \textsf{HSBM}.} }
    \label{Fig:SBM and HSBM}
\end{figure}

\paragraph{Clusters with non-uniform densities.} Next we study graphs in which edges are present \emph{non-uniformly} within each cluster~(e.g., Figure~\ref{Fig:Expande_Clique_Example}(a) discussed earlier). Specifically, we set $k=3, |P_i| = 1000, q=0.002$, $p=0.06$, and  every pair of vertices $(u, v) \in P_i \times P_j$ is connected by an edge with probability $p$ if $i = j$ and probability $q$ otherwise. Moreover,   we choose a random set $S_i\subset P_i$ of size $|S_i| = c_p\cdot |P_i|$ from each cluster, and add edges to connect every pair of vertices in each $S_i$ so that the vertices of each $S_i$ form a clique. By setting different values of $c_p\in[0.05, 0.4]$, the performance of our algorithm is about  $20\%$ -- $50\%$ better than  \texttt{Linkage++}  with respect to the cost value of the constructed tree, see 
Figure~\ref{Fig:SBM Planted Clique}(a) for detailed results. To explain the outperformance of our algorithm, notice that, by adding a clique into some cluster, the cluster structure is usually preserved with respect to   $(\Phi_{\mathrm{in}}, \Phi_{\mathrm{out}})$ or similar eigen-gap  assumptions on well-clustered graphs.  However, the existence of such a clique within some cluster would make the vertices' degrees highly unbalanced; as such many clustering algorithms that involve the matrix perturbation theory in their analysis might not work well.

\begin{figure}[h]
    \centering
    \hspace{-1cm}
    \begin{subfigure}[h]{.43\textwidth}
        \centering
        \scalebox{.25}{\includegraphics{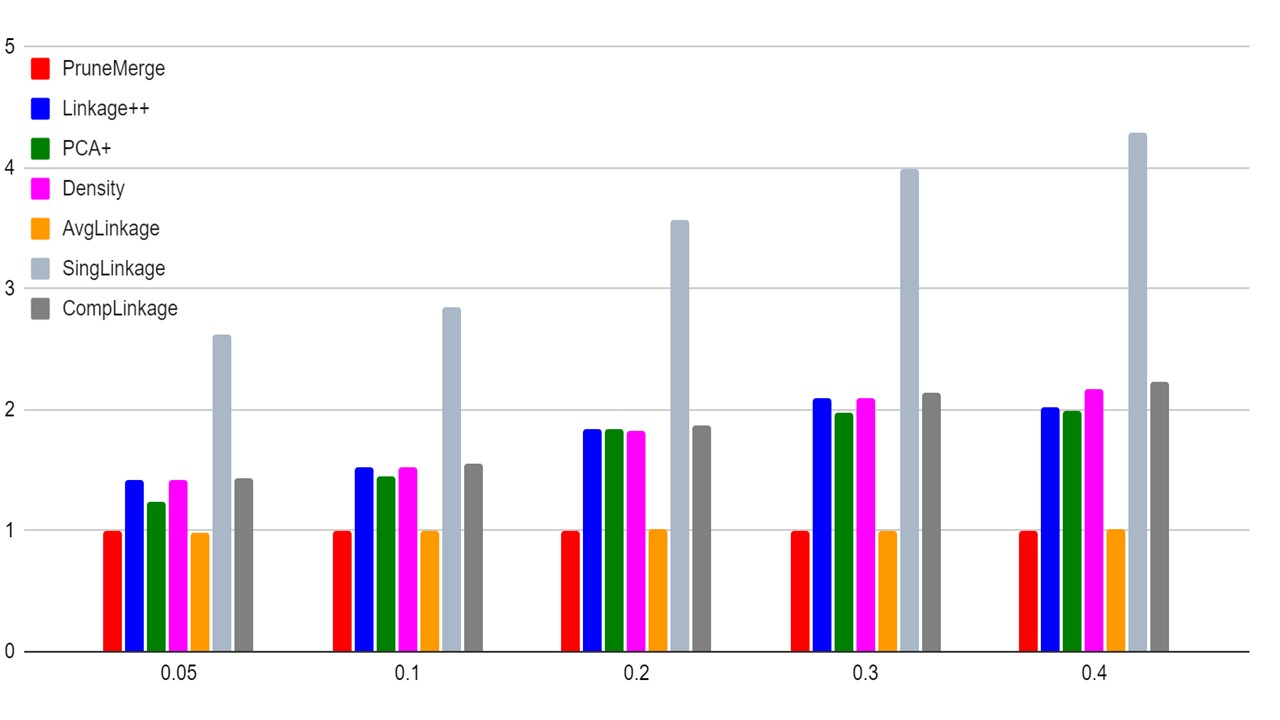}}
    \caption{}
    \label{Fig:SBM Planted Clique Same Sizes}
    \end{subfigure}
\hspace{0.5cm}
    \begin{subfigure}[h]{.43\textwidth}
        \centering
        \scalebox{.25}{\includegraphics{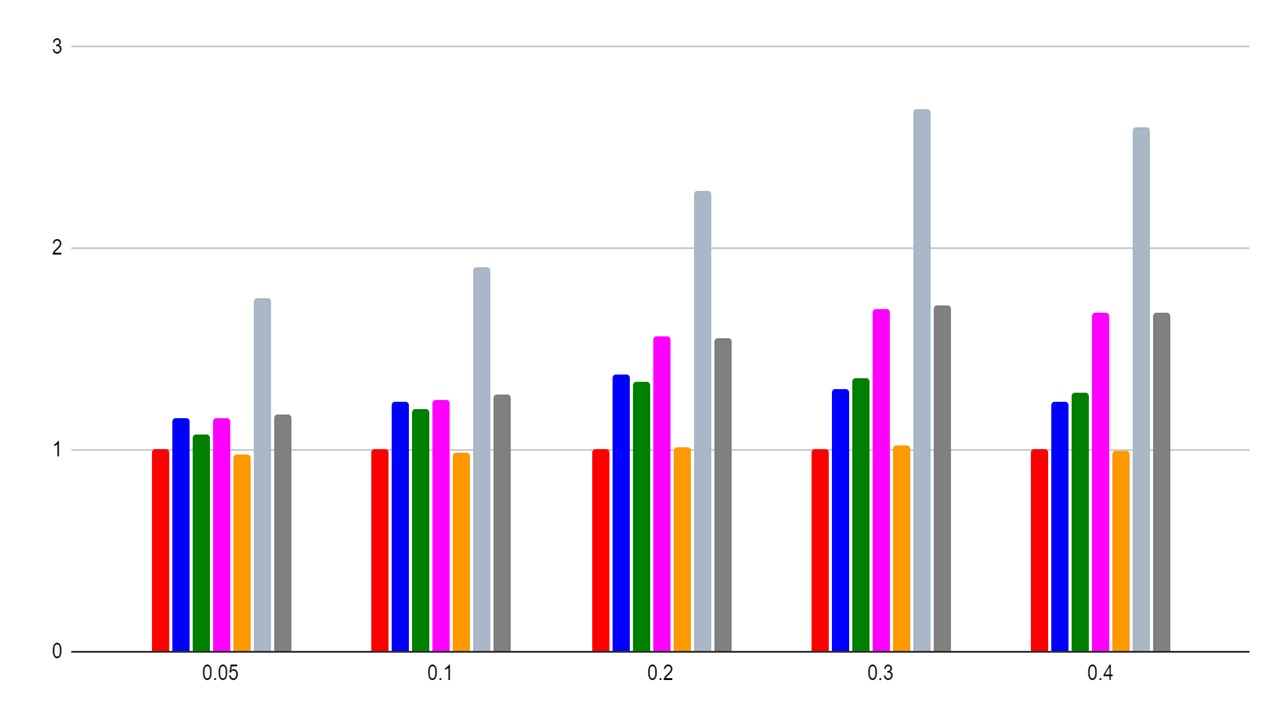}}
    \caption{}
    \label{Fig:SBM Planted Clique Diff Sizes}
    \end{subfigure}
    \caption{\small{Results for graphs with non-uniform densities, or different sizes. The $x$-axis represents the $c_p$-values, while the $y$-axis represents the cost of the algorithms' constructed trees normalised by the cost of the ones constructed by \texttt{PruneMerge}. Figure~(a) corresponds to inputs where all clusters have the same size, while in Figure~(b) the clusters have different sizes. } 
    \label{Fig:SBM Planted Clique}}
\end{figure}

 \paragraph{Clusters of different sizes.} To further highlight the significance of our algorithm on synthetic graphs of non-symmetric structures among the clusters, we study the graphs  in which the clusters have different sizes. We choose the same set of $k$ and $q$ values as before~($k=3, q = 0.002$), but set the sizes of the clusters to be $|P_1| = 1900, |P_2| = 900$ and $|P_3| = 200$. Every pair of vertices $u, v \in P_i$, for $i \in \{1, 2\}$ is connected by an edge with probability $p_1 = 0.06$, while pairs of vertices $u, v \in P_3$ are connected with probability\footnote{Such  choice of $p_2$ is to compensate for the small size of cluster $P_3$, and this ensures that the outer conductance $\Phi_G(P_3)$ is low.} $p_2 = 5 \cdot p_1 = 0.3$. We further plant a clique $S_i \subset P_i$ of size $|S_i| = c_p \cdot P_i$ for each cluster $P_i$, as in the previous set of experiments.   By choosing different values of $c_p$ from $ [0.05, 0.4]$, our results are reported in Figure~\ref{Fig:SBM Planted Clique}(b), demonstrating that our algorithm performs better than the ones in   \cite{CAKMT17}.

\subsection{Experiments on real-world data sets}

To evaluate the performance of our algorithm on real-world data sets, we follow the sequence of recent work on hierarchical clustering~\cite{abboud2019subquadratic,CAKMT17, menon2019online,roy2017hierarchical}, all of which are based on the following 5 data sets from the Scikit-learn library \cite{scikit-learn} as well as the UCI ML repository~\cite{UCIML}: Iris, Wine, Cancer, Boston and Newsgroup\footnote{Due to the very large size of this data set, we consider only a subset consisting of ``comp.graphics'', ``comp.os.ms-windows.misc'', ``comp.sys.ibm.pc.hardware'', ``comp.sys.mac.hardware'', ``rec.sport.baseball'', and ``rec.sport.hockey''.}. Similar with \cite{roy2017hierarchical}, for each data set we construct the similarity graph based on the Gaussian kernel, in which the $\sigma$-value is chosen according to the standard heuristic~\cite{ng2001spectral}. 
 As reported in Figure~\ref{figure:real}, our algorithm performs marginally worse than \texttt{Linkage++} and significantly better than \texttt{PCA+}.

\begin{figure}[h]
    \centering
    \scalebox{.25}{\includegraphics{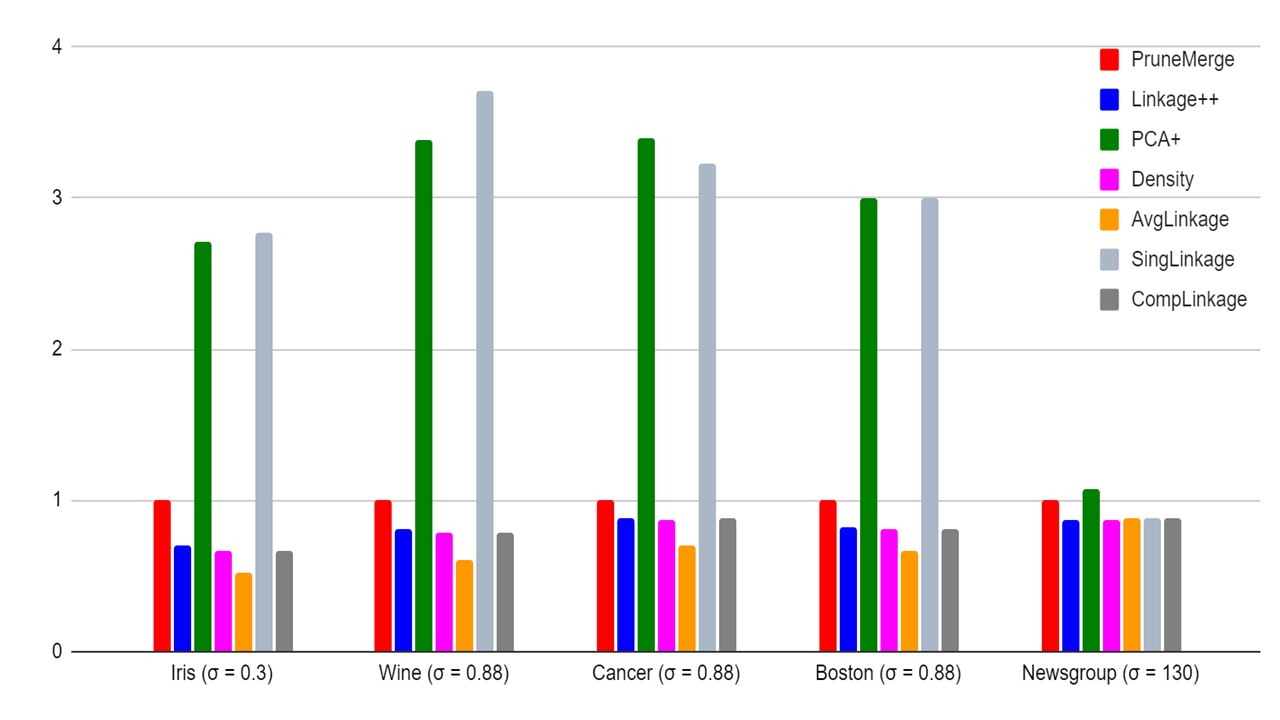}}
    \caption{\small{Results on real-world data sets. The $x$-axis represents the various data sets and our choice of the $\sigma$-value used for constructing the similarity graphs. The $y$-axis corresponds to the cost of the algorithms' output normalised by the cost of  \texttt{PruneMerge}.\label{figure:real}}}
\end{figure}

\section{Conclusion}\label{sec:Conclusion}

Our experimental results on synthetic data sets demonstrate that our presented algorithm \texttt{PruneMerge} not only has excellent theoretical guarantees, but also produces output of lower cost than the previous algorithm \texttt{Linkage++}. In particular, the outperformance of our algorithm is best illustrated on graphs whose clusters have asymmetric internal structure and non-uniform densities. On the other side, our experimental results on real-world data sets show that the performance of \texttt{PruneMerge} is inferior to \texttt{Linkage++} and especially to \texttt{Average Linkage}. We believe that developing more efficient algorithms for well-clustered graphs is a very meaningful direction for future work.

Finally, our  experimental results   indicate that the \texttt{Average Linkage} algorithm performs extremely well on all tested instances, when compared to \texttt{PruneMerge} and \texttt{Linkage++}. This leads to the open question whether \texttt{Average Linkage} achieves an $O(1)$-approximation for well-clustered graphs, although  it fails to achieve this for general graphs~\cite{cohen2018hierarchical}. 
In our point of view, the answer to this question could
help us design more efficient algorithms for hierarchical clustering that not only work in practice, but also have rigorous theoretical guarantees.

\bibliographystyle{alpha}
\bibliography{reference}

\appendix

\end{document}